\documentclass[reprint,superscriptaddress,nofootinbib,amsmath,amssymb,prx]{revtex4-2}

\usepackage{graphicx}
\usepackage[utf8]{inputenc}
\usepackage{colortbl}
\usepackage[colorlinks,
allcolors = blue,
linktocpage = true,
breaklinks]{hyperref}
\usepackage{amsmath,amssymb,amsfonts,amsthm}
\usepackage{ytableau}
\usepackage{pifont}

\usepackage{mathtools}

\usepackage{mathrsfs}

\usepackage{tabularx}
\usepackage{makecell}
\usepackage{multirow}
\usepackage{hhline}
\usepackage{diagbox}
\usepackage{booktabs}
\usepackage{slashbox}
\usepackage{cases}
\usepackage{epigraph} 

\usepackage{tikz}
\usetikzlibrary{external}
\usetikzlibrary{quantikz2}
\usetikzlibrary{backgrounds}
\usetikzlibrary{calc}
\usetikzlibrary{shapes}
\usetikzlibrary{decorations.pathmorphing, decorations.markings}
\usetikzlibrary{decorations.pathreplacing}
\definecolor{darkred}{RGB}{179, 16, 32}

\usepackage{pgfplots}
\pgfplotsset{width=10cm,compat=1.13}

\usepackage{nicefrac}

\usepackage{ytableau,varwidth}

\usepackage{cellspace}
\usepackage{adjustbox}
\usepackage[capitalise]{cleveref}%[nameinlink]         

\usepackage{dsfont}

\makeatletter 
    
\renewcommand\onecolumngrid{% <<<<<<
\do@columngrid{one}{\@ne}%
\def\set@footnotewidth{\onecolumngrid}% <<<<<<<<<<<<<<<<
\def\footnoterule{\kern-6pt\hrule width 1.5in\kern6pt}%
}

\renewcommand\twocolumngrid{% <<<<<<
        \def\footnoterule{% restore rule
        \dimen@\skip\footins\divide\dimen@\thr@@
        \kern-\dimen@\hrule width.5in\kern\dimen@}
        \do@columngrid{mlt}{\tw@}
}%

\makeatother    
\newcommand{\ct}{^{\dagger}}
\newcommand{\tp}{^{\mathsf{T}}}

\newcommand{\x}{\otimes}
\newcommand{\xp}[1]{^{\otimes #1}}

\DeclarePairedDelimiterX\Set[1]\{\}{%
    
    #1
}

\DeclarePairedDelimiter{\set}{\lbrace}{\rbrace}
\DeclarePairedDelimiter{\abs}{\lvert}{\rvert}

\DeclarePairedDelimiter{\of}{\lparen}{\rparen}
\DeclarePairedDelimiter{\sof}{\lbrack}{\rbrack}

\renewcommand{\bra}[1]{\langle{#1}\rvert}
\renewcommand{\ket}[1]{\lvert{#1}\rangle}
\renewcommand{\braket}[2]{\langle{#1}|{#2}\rangle}
\newcommand{\ketbra}[2]{\ket{#1}\bra{#2}}

\newcommand{\defeq}{\vcentcolon=}
\newcommand{\eqdef}{=\vcentcolon}
\renewcommand{\leq}{\leqslant}
\renewcommand{\geq}{\geqslant}

\newcommand{\1}{\mathds{1}}

\newcommand{\C}{\mathbb{C}} % complex numbers
\DeclareMathOperator{\End}{End} % endomorphisms
\DeclareMathOperator{\Tr}{Tr} % trace
\DeclareMathOperator{\U}{U} % unitary group
\DeclareMathOperator{\SU}{SU} % special unitary group
\DeclareMathOperator{\CS}{\mathbb{C}S} % symmetric group algebra
\newcommand{\A}{\mathcal{A}} % algebra of partially transposed permutation matrices
\newcommand{\Irr}[1]{\widehat#1} % set of irreducible modules
\newcommand{\Paths}{\mathrm{Paths}} % set of paths old
\newcommand{\p}[1]{\underline{#1}} % bottom row indices

\newcommand{\I}{\mathcal{I}} % Input of input unitary U
\renewcommand{\O}{\mathcal{O}} % Output of input unitary U
\renewcommand{\P}{\mathcal{P}} % Input of output unitary f(U)
\newcommand{\F}{\mathcal{F}} % Output of output unitary f(U)

\newcommand\restr[2]{{% we make the whole thing an ordinary symbol
  \left.\kern-\nulldelimiterspace % automatically resize the bar with \right
  #1 % the function
  \vphantom{\big|} % pretend it's a little taller at normal size
  \right|_{#2} % this is the delimiter
  }}

\DeclareMathSymbol{\shortminus}{\mathbin}{AMSa}{"39}

\newcommand{\g}{\cellcolor{green!25}}
\newcolumntype{g}{>{\columncolor{green}}c}

\newcommand{\Supp}{\mathrm{Supp}}
\newcommand{\sfT}{\mathsf{T}}
\newcommand{\CC}{\mathbb{C}}

\newcommand{\dket}[1]{\vert {#1} \rangle\!\rangle}
\newcommand{\dbraket}[2]{\langle\!\langle {#1} \vert {#2} \rangle\!\rangle}
\newcommand{\dketbra}[2]{\vert {#1} \rangle\!\rangle\!\langle\!\langle {#2} \vert}

\makeatletter
\def\dbraket#1{%
    \@ifnextchar\bgroup{%
        \dbraket@{#1}%
    }{%
        \langle\!\langle {#1} \vert {#1} \rangle\!\rangle%
    }%
}
\def\dbraket@#1#2{%
    \langle\!\langle {#1} \vert {#2} \rangle\!\rangle%
}
\makeatother

\makeatletter
\def\dketbra#1{%
    \@ifnextchar\bgroup{%
        \dketbra@{#1}%
    }{%
        \vert {#1} \rangle\!\rangle\!\langle\!\langle {#1} \vert%
    }%
}
\def\dketbra@#1#2{%
    \vert {#1} \rangle\!\rangle\!\langle\!\langle {#2} \vert%
}
\makeatother

\newcommand{\dd}{\mathrm{d}}

\newcommand{\ZZ}{\mathbb{Z}}

\newcommand{\mcA}{\mathcal{A}}

\newcommand{\mcC}{\mathcal{C}}

\newcommand{\mcF}{\mathcal{F}}

\newcommand{\mcI}{\mathcal{I}}

\newcommand{\mcM}{\mathcal{M}}

\newcommand{\mcO}{\mathcal{O}}
\newcommand{\mcP}{\mathcal{P}}

\newcommand{\mcR}{\mathcal{R}}

\newcommand{\mcU}{\mathcal{U}}
\newcommand{\mcV}{\mathcal{V}}

\newcommand{\mcX}{\mathcal{X}}

\newcommand{\scB}{\mathscr{B}}

\newcommand{\cellblue}{\cellcolor{blue!25}}
\newcommand{\cellred}{\cellcolor{red!25}}

\newtheorem{definition}{Definition}
\newtheorem{theorem}[definition]{Theorem}

\newtheorem{lemma}[definition]{Lemma}

\begin{document}

\preprint{APS/123-QED}

\title{Sequential quantum processes with group symmetries}
\author{Dmitry Grinko}
\thanks{These authors contributed equally to this work.\\
Dmitry Grinko: \href{mailto:d.grinko@uva.nl}{d.grinko@uva.nl}\\
Satoshi Yoshida: \href{mailto:satoshiyoshida.phys@gmail.com}{satoshiyoshida.phys@gmail.com}}
\affiliation{QuSoft, Amsterdam, The Netherlands}
\affiliation{Institute for Logic, Language and Computation, University of Amsterdam, The Netherlands} 
\affiliation{Korteweg-de Vries Institute for Mathematics, University of
Amsterdam, The Netherlands}
\author{Satoshi Yoshida}
\thanks{These authors contributed equally to this work.\\
Dmitry Grinko: \href{mailto:d.grinko@uva.nl}{d.grinko@uva.nl}\\
Satoshi Yoshida: \href{mailto:satoshiyoshida.phys@gmail.com}{satoshiyoshida.phys@gmail.com}}
\affiliation{Department of Physics, Graduate School of Science, The University of Tokyo, Hongo 7-3-1, Bunkyo-ku, Tokyo 113-0033, Japan}
\author{Mio Murao}
\affiliation{Department of Physics, Graduate School of Science, The University of Tokyo, Hongo 7-3-1, Bunkyo-ku, Tokyo 113-0033, Japan}
\affiliation{Trans-scale Quantum Science Institute, The University of Tokyo, Bunkyo-ku, Tokyo 113-0033, Japan}
\author{Maris Ozols}
\affiliation{QuSoft, Amsterdam, The Netherlands}
\affiliation{Institute for Logic, Language and Computation, University of Amsterdam, The Netherlands} 
\affiliation{Korteweg-de Vries Institute for Mathematics, University of
Amsterdam, The Netherlands}

\date{\today}

\begin{abstract}
    Symmetry plays a crucial role in the design and analysis of quantum protocols.
    This result shows a canonical circuit decomposition of a $(G\times H)$-invariant quantum comb for compact groups $G$ and $H$ using the corresponding Clebsch--Gordan transforms, which naturally extends to the $G$-covariant quantum comb.
    By using this circuit decomposition, we propose a parametrized quantum comb with group symmetry, and derive the optimal quantum comb which transforms an unknown unitary operation $U\in \SU(d)$ into its inverse $U^\dagger$ or transpose $U\tp$.
    From numerics, we find a deterministic and exact unitary transposition protocol for $d=3$ with $7$ queries to $U$.
    This protocol improves upon the protocol shown in the previous work, which requires $13$ queries to $U$.
\end{abstract}

\maketitle

\section{Introduction}

Symmetry plays an essential role in analyzing quantum systems and quantum information protocols~\cite{noether1971invariant,itzykson1966unitary, harrow2005applications, hayashi2017group1, hayashi2017group2}.
Quantum cloning is one of the most studied protocols in this context, as it highlights the trade-offs between accuracy and resource consumption in quantum state manipulation~\cite{wootters1982single, dieks1982communication, barnum1996noncommuting, buvzek1996quantum, gisin1997optimal, bruss1998optimal, werner1998optimal, keyl1999optimal,cwiklinski2012region, studzinski2014group, AsymmetricCloningCerf, scarani2005quantum,dariano2003optimal, AsymmetricCloning, nechita2023asymmetric}.
For $m\to n$ cloners, since the quantum state $U^{\otimes m}\ket{\psi}^{\otimes m}$ should be ideally cloned to a quantum state $U^{\otimes n}\ket{\psi}^{\otimes n}$ for any $U\in \U(d)$ (which is, of course, not possible exactly for $m<n$) and input and output states are invariant under the permutation, the optimal cloning channel satisfies the unitary group covariance and the permutation invariance.
In this way, its optimal performance and corresponding protocol are determined.
Beyond quantum cloning, symmetry is used in the analysis of various scenarios such as open quantum systems~\cite{marvian2014extending,cirstoiu2020robustness}, quantum Shannon theory~\cite{werner2002counterexample,king2003capacity,holevo2005additivity,datta2006additivity,koenig2009strong,datta2016second,wilde2017converse,pirandola2017fundamental,brannan2020temperley}, quantum error correction~\cite{kong2021nearoptimal,lin2025covariant}, and quantum machine learning~\cite{glick2024covariant}.
References~\cite{grinko2024linear,grinko2023gelfand,grinko2025mixed} provide a systematic way to utilize the unitary group symmetry to analyze the semidefinite programming, which often appears as the optimization problem in quantum information processing~\cite{chiribella2016optimal, tavakoli2024semidefinite}.
The quantum Schur transforms and mixed Schur transforms~\cite{harrow2005applications, bacon2006efficient, bacon2007quantum, nguyen2023mixed, grinko2023gelfand, fei2023efficient, grinko2025mixed} provide an efficient implementation of several protocols with the unitary group symmetry (e.g., Refs.~\cite{buhrman2022quantum, cervero2023weak, cervero2024memory}).

Recently, the quantum supermaps have been developed to analyze the transformation of quantum channels, such as dynamical resource theory~\cite{chitambar2019quantum} and higher-order quantum transformation~\cite{taranto2025higher}.
The quantum supermaps are implemented by the quantum combs in the quantum circuit, which takes the input state on stream and outputs states accordingly~\cite{chiribella2008qca}.
This circuit structure also describes non-Markovian dynamics~\cite{pollock2018non}, and it is also related to the causal structure~\cite{oreshkov2012quantum}.
The quantum supermaps can be represented as a Choi matrix, and their optimization is often formulated as semidefinite programming~\cite{chiribella2016optimal}.
Thus, unitary group symmetry is helpful in analyzing its optimal performance~\cite{grinko2024linear}.
However, its efficient description in a quantum circuit is not known.

This work shows a canonical circuit decomposition of the $\U(d)\times \U(d)$-invariant quantum comb.
This utilizes the (dual) Clebsch--Gordan transforms~\cite{harrow2005applications, bacon2006efficient,bacon2007quantum, nguyen2023mixed, grinko2023gelfand, fei2023efficient} combined with an arbitrary isometry operation acting on the multiplicity space.
This construction is extended to ($G\times H$)-invariant quantum comb for compact groups $G$ and $H$ by using the generalized Clebsch--Gordan transforms, which is further extended to the $G$-covariant quantum comb.
By combining this circuit structure with the parametrized quantum comb proposed in Ref.~\cite{mo2025parameterized}, we propose a parameterized quantum comb with group symmetry, which significantly reduces the number of variables to be optimized.
This parameterized quantum comb is used to derive the optimal quantum comb for unitary inversion and transposition, which are the tasks to transform $n$ copies of an unknown unitary operation $U\in \SU(d)$ into $U\ct$ and $U\tp$, respectively~\cite{chiribella2016optimal, sardharwalla2016universal, quintino2019probabilistic, quintino2019reversing, dong2021success, quintino2022deterministic, yoshida2023reversing, navascues2018resetting, trillo2020translating, trillo2023universal, chen2024quantum, mo2025parameterized}.
We numerically show that qutrit-unitary transposition is possible with $7$ queries of the input unitary $U$.
This is improved over Ref.~\cite{chen2024quantum} showing qutrit-unitary transposition with $13$ queries.

The rest of this paper is organized as follows.
Section \ref{sec:preliminaries} introduces the necessary background on quantum combs and group theory.
Section \ref{sec:main_results} presents the main results of this work, which are the canonical circuit decomposition of the $(G\times H)$-invariant quantum comb and its extension to the $G$-covariant quantum comb.
Section \ref{sec:application} presents the application of the main results, which provides an efficient parametrization of the quantum comb with group symmetry and its application to unitary inversion and transposition.
Appendix~\ref{sec:proof_main_results} provides the proof of the main results shown in Sec.~\ref{sec:main_results}.
Section \ref{sec:conclusion} concludes the paper.

\begin{figure}
    \centering
    \includegraphics[width=\linewidth]{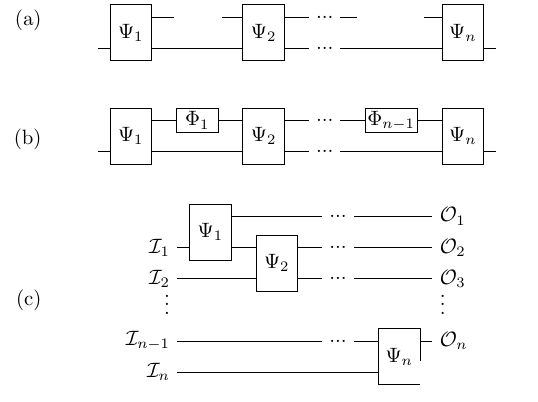}
    \caption{(a) Quantum comb with $n-1$ open slots, where $\Psi_1, \ldots, \Psi_{n}$ are quantum channels.\\
    (b) Quantum comb transforms quantum channels $\Phi_1, \ldots, \Phi_{n-1}$ into a quantum channel $\Phi_\mathrm{out}$ given in Eq.~\eqref{eq:comb_output}.\\
    (c) Quantum comb can be considered as a quantum channel $\Psi: \bigotimes_{i=1}^{n} \End(\mcI_i) \to \bigotimes_{i=1}^{n} \End(\mcO_i)$ with non-signaling conditions $\mcI_i \not\to \mcO_j$ for $j<i$.}
    \label{fig:quantum_comb}
\end{figure}

\section{Preliminaries}
\label{sec:preliminaries}
In this section, we introduce the necessary background on quantum combs and group theory.
Section~\ref{sec:quantum_comb} introduces the definition of quantum combs and their Choi matrix representation.
Section~\ref{sec:commutant_tensor_representation} introduces the commutant of the tensor representation of compact groups, which is used to analyze the quantum combs with group symmetries.
\subsection{Quantum comb}
\label{sec:quantum_comb}
A quantum comb is a quantum circuit composed of $n$ quantum channels $\Psi_i: \End(\mcI_i \otimes \mcA_{i-1}) \to \End(\mcO_i \otimes \mcA_i)$ for $i\in [n]$ with open slots in between~\cite{chiribella2008qca}, where $\mcI_i$, $\mcO_i$, and $\mcA_i$ are the input, output, and auxiliary systems of the channel $\Psi_i$, respectively, $\mcA_0\simeq \mcA_{n} \simeq \CC$ are trivial systems, $\End(\mcX)$ represents the set of linear operators on a Hilbert space $\mcX$, and $[n] \defeq \{1, \ldots, n\}$.
It defines a quantum supermap that transforms $n-1$ quantum channels $\Phi_i: \End(\mcO_i) \to \End(\mcI_{i+1})$ for $i\in [n-1]$ into a quantum channel $\Phi_\mathrm{out}: \End(\mcI_1) \to \End(\mcO_n)$ given by
\begin{align}
    \label{eq:comb_output}
    \Phi_\mathrm{out}
    &= \mathcal{C}(\Phi_1, \ldots, \Phi_{n-1}) \\
    & \defeq \Psi_{n} \circ (\Phi_{n-1} \otimes \1_{\mcA_{n-1}}) \circ \cdots \circ (\Phi_1 \otimes \1_{\mcA_1}) \circ \Psi_1,
\end{align}
where $\1_{\mcA_i}$ is the identity channel on the auxiliary system $\mcA_i$.
It can also be seen as a quantum channel $\Psi: \bigotimes_{i=1}^{n} \End(\mcI_i) \to \bigotimes_{i=1}^{n} \End(\mcO_i)$ with non-signaling conditions $\mcI_i \not\to \mcO_j$ for $j<i$~\cite{eggeling2002semicausal}.
Figure~\ref{fig:quantum_comb} illustrates these three equivalent representations of a quantum comb.
The quantum comb can be represented in a Choi matrix given by
\begin{align}
    C \defeq \sum_{\vec{i}, \vec{j}} \ketbra{\vec{i}}{\vec{j}}_{\mcI^n} \otimes \Psi(\ketbra{\vec{i}}{\vec{j}})_{\mcO^n} \in \End(\mcI^n\otimes \mcO^n),
\end{align}
where $\mcI^n$ and $\mcO^n$ are joint Hilbert spaces given by $\mcI^n \defeq \bigotimes_{i=1}^n \mcI_i$ and $\mcO^n \defeq \bigotimes_{i=1}^n \mcO_i$, respectively, and $\{\ket{\vec{i}}\}$ is an orthonormal basis of $\mcI^n$.
The Choi matrix $C$ satisfies the following constraints~\cite{chiribella2008qca}:
\begin{align}
    \label{eq:comb_condition}
    \begin{split}
    C &\succeq 0,\\
    \Tr_{\mcO_i} C_i &= C_{i-1} \otimes \1_{\mcI_{i}} \quad \forall i \in [n],\\
    C_0 &= 1,
    \end{split}
\end{align}
where $C_{n} \defeq C$ and $C_{i-1} \defeq \frac{1}{d}\Tr_{\mcI_i \mcO_{i}} C_i$.
Conversely, any matrix $C$ satisfying the comb condition~\eqref{eq:comb_condition} is the Choi matrix of a quantum comb~\cite{chiribella2008qca}.

\subsection{Commutant of the tensor representation of compact groups}
\label{sec:commutant_tensor_representation}

We consider unitary representations $\rho_i: G\to \End(V_{\rho_i})$ of a compact group $G$ for $i\in [N]$, where $V_{\rho_i}$ is a representation space and $\End(V_{\rho_i})$ represents the group of invertible operators on $V_{\rho_i}$.
In this section, we investigate the commutant of the tensor representation $\bigotimes_{i=1}^N \rho_i$ defined by:
\begin{align}
    &\mathrm{Comm}\left(\otimes_{i=1}^n \rho_i\right)\nonumber\\
    &\defeq \left\{X\in \End\left(\otimes_{i=1}^n V_{\rho_i}\right) \; \middle| \; \left[X, \otimes_{i=1}^n \rho_i(g)\right] = 0 \quad \forall g\in G\right\}.
    \label{eq:commutant_definition}
\end{align}
Due to the Peter--Weyl theorem, the tensor representation $\bigotimes_{i=1}^n \rho_i$ is decomposed into irreducible representations (irreps) as
\begin{align}
    \bigotimes_{i=1}^n V_{\rho_i} &\cong \bigoplus_{\lambda\in \Irr{G}^{(n)}} V_\lambda \otimes \CC^{M_\lambda^{(n)}},\\
    \bigotimes_{i=1}^n \rho_i &\cong \bigoplus_{\lambda\in \Irr{G}^{(n)}} \lambda(g) \otimes \1_{M_\lambda^{(n)}} \quad \forall g\in G,
\end{align}
where $\Irr{G}^{(n)}$ is the set of irreps of $G$ appearing in the decomposition of $\bigotimes_{i=1}^n \rho_i$, $V_\lambda$ is the representation space of $\lambda\in \Irr{G}^{(n)}$, and $m_\lambda^{(i)}$ is the multiplicity of $\lambda$ in $\rho_i$.
Therefore, due to Schur's lemma, the commutant is spanned by operators given by
\begin{align}
\label{eq:def_E}
    E^\lambda_{pq} &\defeq \left(U_\mathrm{Sch}^{\rho_1, \ldots, \rho_n}\right)^\dagger \1_{V_\lambda} \otimes \ketbra{p}{q}_{\CC^{M_\lambda^{(n)}}} \left(U_\mathrm{Sch}^{\rho_1, \ldots, \rho_n}\right),
\end{align}
where $\{\ket{p}\}$ is an orthonormal basis of $\CC^{M_\lambda^{(n)}}$.

The multiplicity space $\CC^{M_\lambda^{(n)}}$ is recursively determined as follows.
Due to the Peter--Weyl theorem, the representation $\rho_i$ is decomposed into irreps as
\begin{align}
    V_{\rho_i} &\cong \bigoplus_{\nu\in \Irr{G}_i} V_\nu \otimes \CC^{m_\nu^{(i)}},\\
    \rho_i &\cong \bigoplus_{\nu\in \Irr{G}_i} \nu(g) \otimes \1_{m_{\rho_i}^{\nu}} \quad \forall g\in G,
\end{align}
where $\Irr{G}_i$ is the set of irreps of $G$ appearing in the decomposition of $\rho_i$ and $m_{\rho_i}^{\nu}$ is the multiplicity of the irrep $\nu$ in $\rho_i$.
Therefore, for $n=1$, we have $M_\lambda^{(1)} = m_{\rho_1}^\lambda$.
For $n\geq 2$, we have
\begin{align}
    &\bigoplus_{\lambda_n \in \Irr{G}^{(n)}} V_{\lambda_n} \otimes \CC^{M_{\lambda_n}^{(n)}} \cong \bigotimes_{i=1}^n V_{\rho_i}\\
    &\cong \left(\bigotimes_{i=1}^{n-1} V_{\rho_i}\right) \otimes V_{\rho_n}\\
    &\cong \bigoplus_{\lambda_{n-1} \in \Irr{G}^{(n-1)}} V_{\lambda_{n-1}} \otimes \CC^{M_{\lambda_{n-1}}^{(n-1)}} \otimes \bigoplus_{\nu \in \Irr{G}_n} V_\nu \otimes \CC^{m_{\rho_n}^{\nu}}.
\end{align}
Due to the Peter--Weyl theorem, the tensor product representation $V_{\lambda_{n-1}} \otimes V_\nu$ is decomposed into irreps as
\begin{align}
    V_{\lambda_{n-1}} \otimes V_\nu \cong \bigoplus_{\lambda_n\in \Irr{G}^{(n)}} V_{\lambda_n} \otimes \CC^{c_{\lambda_{n-1}\nu}^\lambda},
\end{align}
where $c_{\lambda_{n-1}\nu}^{\lambda_{n}} \in \ZZ_{\geq 0}$ is the multiplicity of $\lambda$ in $V_{\lambda_{n-1}} \otimes V_\nu$ called the Littlewood--Richardson coefficient.
By using the Schur orthogonality relation of the characters of irreps, the Littlewood--Richardson coefficient for irreps $\mu, \nu, \lambda\in \Irr{G}$ is given by
\begin{align}
    \label{eq:littlewood_richardson}
    c_{\mu\nu}^\lambda = \int_G \chi_\mu(g) \chi_\nu(g) \overline{\chi_\lambda(g)} \dd g,
\end{align}
where $\chi_\mu$ is the character of the representation $\mu$ given by $\chi_\mu(g) = \Tr[\rho_\mu(g)]$ and $\dd g$ is the Haar measure of $G$.
For later convenience, we also introduce the irreducible decomposition of $V_{\lambda_{n-1}} \otimes V_{\rho_n}$ given by
\begin{align}
    &V_{\lambda_{n-1}} \otimes V_{\rho_n}
    \cong V_{\lambda_{n-1}} \otimes \bigoplus_{\nu \in \Irr{G}_n} V_\nu \otimes \CC^{m_{\rho_n}^{\nu}}\\
    &\cong \bigoplus_{\lambda_n\in \Irr{G}^{(n)}} V_{\lambda_{n}} \otimes \left[\bigoplus_{\nu\in \Irr{G}_n} \CC^{m_{\rho_n}^{\nu}} \otimes \CC^{c_{\lambda_{n-1}\nu}^{\lambda_n}}\right]\\
    &= \bigoplus_{\lambda_n\in \Irr{G}^{(n)}} V_{\lambda_{n}} \otimes \CC^{c_{\lambda_{n-1} \rho_n}^{\lambda_n}},
    \label{eq:decomposition_tensor_representation}
\end{align}
where $c_{\lambda_{n-1} \rho_n}^{\lambda_n}$ is defined by
\begin{align}
    c_{\lambda_{n-1} \rho_n}^{\lambda_n} \defeq \sum_{\nu\in \Irr{G}_n} m_{\rho_n}^{\nu} c_{\mu\nu}^\lambda.
\end{align}
We also introduce the notation
\begin{align}
    \Irr{G}^{(0)} = \{\varnothing\},\quad V_\varnothing \defeq \CC, \quad c_{\varnothing \rho_1}^\lambda \defeq m_{\rho_1}^{\lambda}
\end{align}
such that Eq.~\eqref{eq:decomposition_tensor_representation} holds for $n=1$.
Therefore, we have
\begin{align}
    \CC^{M_\lambda^{(n)}}
    &\cong \bigoplus_{\mu\in \Irr{G}^{(n-1)}} \CC^{M_\mu^{(n-1)}} \otimes \CC^{c_{\lambda_{n-1} \rho_n}^{\lambda_n}}.
    \label{eq:decomp:multiplicity}
\end{align}
Applying this decomposition recursively for $n$, we obtain
\begin{align}
    \CC^{M_\lambda^{(n)}}
    &\cong \bigoplus_{\substack{\lambda_1\in \Irr{G}_1, \lambda_2\in \Irr{G}^{(2)},\\ \ldots, \lambda_{n-1} \in \Irr{G}^{(n-1)}}} \CC^{c_{\varnothing \rho_1}^{\lambda_1}} \otimes \CC^{c_{\lambda_1 \rho_2}^{\lambda_2}}\otimes \cdots \CC^{c_{\lambda_{n-1} \rho_n}^{\lambda}}.
\end{align}

\begin{figure}
    \centering
    \includegraphics[width=\linewidth]{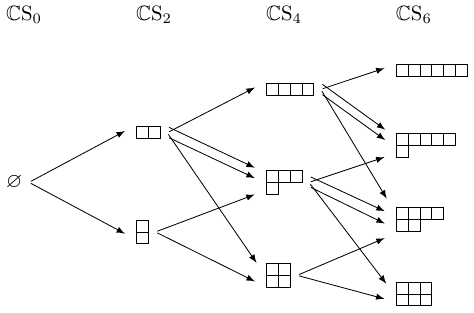}
    \caption{Bratteli diagram for the commutant algebra chain $\CS_0 \hookrightarrow \CS_2 \hookrightarrow \CS_4 \hookrightarrow \CS_6$, with each commutant defined as in Eq.~\eqref{eq:commutant_definition}, associated with the tensor representations $\rho_1(g) = \rho_2(g) = \rho_3(g) = g^{\otimes 2}$ for $g\in G = \U(2)$.
    Vertices at level $i$ are the irreps appearing in the commutant of $g^{\otimes 2i}$, and parallel edges indicate multiplicities greater than one.}
    \label{fig:bratelli_diagram}
\end{figure}

To describe this decomposition, we introduce a Bratteli diagram $\scB$, which is a graph composed of vertices labelled by $\ZZ_{\geq 0}$ (called ``levels''), as follows.
The set of vertices at level $0$, denoted by $V_0$, is given by $\{\varnothing\}$.
The set of vertices at level $i$ ($1\leq i\leq n$), denoted by $V_i$, is given by
\begin{align}
    V_i\defeq \{\lambda_i \mid \exists \lambda_{i-1}\in V_{i-1}, c_{\lambda_{i-1} \rho_i}^{\lambda_i}>0\}.
\end{align}
The set of edges from level $i-1$ to level $i$, denoted by $E_i$, is given by
\begin{align}
    E_i\defeq \{ \lambda_{i-1} \xrightarrow{a_i} \lambda_i \mid \lambda_{i-1} \in V_{i-1}, \lambda_i \in V_i, a_i\in [c_{\lambda_{i-1} \rho_i}^{\lambda_i}] \}.
\end{align}
As an explicit example of this definition, Fig.~\ref{fig:bratelli_diagram} shows the Bratteli diagram for the commutant algebra chain $\CS_0 \hookrightarrow \CS_2 \hookrightarrow \CS_4 \hookrightarrow \CS_6$ associated with $\rho_1(g) = \rho_2(g) = \rho_3(g) = g^{\otimes 2}$ for $g\in G = \U(2)$.
Then, the multiplicity space $\CC^{M_\lambda^{(n)}}$ is spanned by an orthonormal basis labeled by the set of paths from $\varnothing\in V_0$ to $\lambda\in V_n$, denoted by $\Paths(\lambda, \scB)$:
\begin{align}
    &\Paths(\lambda, \scB)\nonumber\\
    &\defeq \left \{\begin{aligned}
    &\varnothing \xrightarrow{a_1} \lambda_1 \xrightarrow{a_2}\cdots\\
    &\xrightarrow{a_{n-1}} \lambda_{n-1} \xrightarrow{a_{n}} \lambda
    \end{aligned} \;\middle|\;
    \begin{aligned}
    &\lambda_1\in V_1, \ldots, \lambda_{n-1}\in V_{n-1}, \\
    &a_1\in [c_{\varnothing \rho_1}^{\lambda_1}], \ldots, a_n\in [c_{\lambda_{n-1} \rho_n}^{\lambda}]
    \end{aligned} \right\}.
\end{align}
The corresponding basis vector is recursively given by
\begin{align}
    \label{eq:path_vector}
    \begin{split}
    &\ket{\varnothing \xrightarrow{a_1} \lambda_1}\defeq \ket{a_1},\\
    &\ket{\varnothing \xrightarrow{a_1} \lambda_1 \xrightarrow{a_2} \cdots \xrightarrow{a_{n-1}} \lambda_{n-1} \xrightarrow{a_{n}} \lambda_{n}}\\
    &\defeq \ket{\varnothing \xrightarrow{a_1} \lambda_1 \xrightarrow{a_2}  \cdots \xrightarrow{a_{n-1}} \lambda_{n-1}} \otimes \ket{a_n},
    \end{split}
\end{align}
where $\{\ket{a_1}\}_{a_1\in [c_{\varnothing\rho_1}^{\lambda_1}]}$ is an orthonormal basis of $\CC^{c_{\varnothing\rho_1}^{\lambda_1}}$, $\ket{\varnothing \xrightarrow{a_1} \lambda_1 \xrightarrow{a_2} \cdots \xrightarrow{a_{n-1}} \lambda_{n-1}}$ is a basis vector of $\CC^{M_{\lambda_{n-1}}^{(n-1)}}$, and $\{\ket{a_n}\}_{a_n\in [c_{\lambda_{n-1} \rho_n}^{\lambda}]}$ is an orthonormal basis of $\CC^{c_{\lambda_{n-1} \rho_n}^{\lambda}}$ in the decomposition~\eqref{eq:decomp:multiplicity}.
Using this orthonormal basis of the multiplicity space $\CC^{M_\lambda^{(n)}}$, we define an orthogonal basis $\{E^{\lambda}_{pq} \mid \lambda \in \Irr{G}, p,q\in \Paths(\lambda, \scB)\}$ of $\mathrm{Comm}\left(\bigotimes_{i=1}^n \rho_i\right)$ by Eq.~\eqref{eq:def_E}, i.e.,
\begin{align}
    &\mathrm{Comm}\left(\otimes_{i=1}^n \rho_i\right) \nonumber\\
    &= \mathrm{span}\left\{E^\lambda_{pq} \;\middle | \; \lambda \in \Irr{G}, p,q\in \Paths(\lambda, \scB) \right\}.
\end{align}
Then, the orthogonal basis $\{E^{\lambda}_{pq} \mid \lambda \in \Irr{G}^{(n)}, p,q\in \Paths(\lambda, \scB)\}$ satisfies the following properties:
\begin{lemma}
    \label{lem:orthogonal_basis_properties}
    The orthogonal basis $\{E^{\lambda}_{pq} \mid \lambda \in \Irr{G}^{(n)}, p,q\in \Paths(\lambda, \scB)\}$ satisfies
    \begin{align}
        E^\lambda_{pq} E^{\lambda'}_{p'q'} &= \delta_{\lambda\lambda'} \delta_{q p'} E^\lambda_{p q'},\\
        (E^\lambda_{pq})^\dagger &= E^\lambda_{qp},\\
        \Tr E^\lambda_{pq} &= \delta_{pq} d_\lambda,
    \end{align}
    where $\delta_{\lambda\lambda'}$ is 1 if $\lambda=\lambda'$ and 0 otherwise, and $d_\lambda\coloneqq \dim V_\lambda$.
\end{lemma}
\begin{proof}
    The first two properties follow from the definition~\eqref{eq:def_E}.
    The last property is shown as follows:
    \begin{align}
        \Tr E^\lambda_{pq} &= \Tr(\1_{V_\lambda} \otimes \ketbra{p}{q})\\
        &= \Tr(\1_{V_\lambda}) \Tr(\ketbra{p}{q})\\
        &= \delta_{pq} d_\lambda.
    \end{align}
\end{proof}
\begin{lemma}
    \label{lem:tensor_and_partial_trace}
    The orthogonal bases $\{E^{\lambda}_{pq} \mid \lambda \in \Irr{G}^{(n)}, p,q\in \Paths(\lambda, \scB)\}$ and $\{E^{\mu}_{rs} \mid \mu \in \Irr{G}^{(n-1)}, r,s\in \Paths(\mu, \scB)\}$ satisfy
    \begin{align}
    \label{eq:E_tensor_I}
        &E^\mu_{rs} \otimes \1_{d_n} \nonumber\\
        &= \sum_{\lambda\in \Irr{G}^{(n)}} \sum_{p, q\in \Paths(\lambda, \scB)} \sum_{a\in [c_{\mu\rho_n}^{\lambda}]} \delta_{p, (r\xrightarrow{a}\lambda)} \delta_{q, (s\xrightarrow{a}\lambda)} E^\lambda_{pq},\\
    \label{eq:E_partial_trace}
        &\Tr_n E^\lambda_{pq} \nonumber\\
        &= {d_\lambda \over d_\mu}\sum_{\mu\in \Irr{G}^{(n-1)}} \sum_{r, s\in \Paths(\mu, \scB)} \sum_{a\in [c_{\mu\rho_n}^{\lambda}]} \delta_{p, (r\xrightarrow{a}\lambda)} \delta_{q, (s\xrightarrow{a}\lambda)} E^\mu_{rs},
    \end{align}
    where $r\xrightarrow{a}\lambda$ represents the path obtained by adding the edge $a$ from $\mu$ to $\lambda$ to the path $r$.
\end{lemma}

We show the proof of Lem.~\ref{lem:tensor_and_partial_trace} in Supplemental Material~\ref{subsec:proof_lem_tensor_and_partial_trace}.

\section{Main results}
\label{sec:main_results}
This section presents the main results of this work.
As a warmup, we first consider the implementation of a $\U(d)$-invariant quantum channel using the quantum Schur transform in Sec.~\ref{subsec:warmup}.
Then, we present the canonical circuit decomposition of the $(G\times H)$-invariant quantum comb for compact groups $G$ and $H$ in Sec.~\ref{subsec:canonical_decomposition}.
Finally, we present the canonical circuit decomposition of the $G$-covariant quantum comb for a compact group $G$ in Sec.~\ref{subsec:covariant_comb}.

\subsection{Warmup: Implementation of $\U(d)$-invariant quantum channel using quantum Schur transform}
\label{subsec:warmup}
We consider a quantum channel $\Psi: \bigotimes_{i=1}^{n} \End(\mcI_i) \to \bigotimes_{i=1}^{n} \End(\mcO_i)$ that is invariant under the action of the unitary group $\U(d)$, i.e.,
\begin{align}
    \Psi(U^{\otimes n} \cdot U^{\otimes n \dagger}) = V^{\otimes n} \Psi(\cdot) V^{\otimes n \dagger} \quad \forall U, V \in \U(d).
\end{align}
The Choi matrix $C$ of the channel $\Psi$ satisfies
\begin{align}
    \label{eq:choi_sud_invariant}
    [C, U^{\otimes n}_{\mcI^n} \otimes V^{\otimes n}_{\mcO^n}] = 0 \quad \forall U, V \in \U(d).
\end{align}
This channel $\Psi$ can be implemented by using the quantum Schur transform~\cite{harrow2005applications,bacon2006efficient,bacon2007quantum}, defined as follows.
The tensor product $U^{\otimes n}$ can be considered as a representation of $\U(d)$ on the Hilbert space $(\CC^d)^{\otimes n}$, which is decomposed into irreducible representations (irreps) as
\begin{align}
    \label{eq:sw_duality_space}
    (\CC^d)^{\otimes n} &\cong \bigoplus_{\lambda \vdash_d n} V_\lambda \otimes \CC^{m_{\lambda}},\\
    U^{\otimes n} &\cong \bigoplus_{\lambda \vdash_d n} U_\lambda \otimes \1_{m_{\lambda}},
\end{align}
where $\lambda \vdash_d n$ means that $\lambda$ is a partition of $n$ into at most $d$ parts, $U_\lambda$ is an irrep of $\U(d)$ labeled by $\lambda$, $m_\lambda$ is the multiplicity of the irrep $U_\lambda$, $V_\lambda$ is the representation space corresponding to $U_\lambda$, and $\1_{m_\lambda}$ is the identity operator on the multiplicity space $\CC^{m_{\lambda}}$.
The quantum Schur transform $U_{\text{Sch}}^{(n)}$ is a unitary operator that transforms the computational basis $\ket{\vec{i}}$ of $(\CC^d)^{\otimes n}$ into the Schur basis $\ket{\lambda, p_\lambda, \psi_\lambda}$, where $\lambda$ labels the irrep $U_\lambda$, $\psi_\lambda$ labels the basis of the irrep space, and $p_\lambda$ labels the basis of the multiplicity space.
In the quantum circuit, the quantum Schur transform $U_{\text{Sch}}^{(n)}$ is represented as an isometry operator
\begin{align}
    U_{\text{Sch}}^{(n)}: (\CC^d)^{\otimes n} \to \mcR_n \otimes \mcP_n \otimes \mcV_n,
\end{align}
where $\mcR_n$ stores the irreps $\lambda\vdash_d n$, $\mcP_n$ stores the vectors $\ket{p_\lambda}\in \CC^{m_{\lambda}}$ for irreps $\lambda\vdash_d n$, and $\mcV_n$ stores the vectors $\ket{\psi_\lambda}\in V_\lambda$ for irreps $\lambda\vdash_d n$.
Due to Schur's lemma, the Choi matrix $C$ satisfying Eq.~\eqref{eq:choi_sud_invariant} is block-diagonalized in the Schur basis as
\begin{align}
    C\cong \bigoplus_{\lambda, \mu \vdash_d n} \1_{V_\lambda} \otimes \1_{V_\mu} \otimes C^{\lambda \mu}
\end{align}
using $C^{\lambda \mu} \in \End(\CC^{m_\lambda} \otimes \CC^{m_{\mu}})$.
Thus, the channel $\Psi$ can be written as
\begin{align}
    \Psi(\cdot) = U_\mathrm{Sch}^{(n)\dagger} \left[\bigoplus_{\lambda, \mu \vdash_d n} \pi_{V_\mu} \otimes \Psi^{\lambda\mu} \circ \Tr_{V_\lambda} \left( \Pi_\lambda \cdot \right)\right] U_\mathrm{Sch}^{(n)},
\end{align}
where $\Psi^{\lambda \mu}: \End(\CC^{m_{\lambda}}) \to \End(\CC^{m_{\mu}})$, $\pi_{V_\mu}$ is the maximally mixed state defined by $\pi_{\mu}\defeq \1_{V_\mu}/d_\mu$, and $\Pi_\lambda$ is the Young projector defined by $\Pi_\lambda \defeq U_\mathrm{Sch}^{(n)\dagger}(\1_{V_{\lambda}} \otimes \1_{m_\lambda}) U_\mathrm{Sch}^{(n)}$.
The channel $\Psi$ can be implemented by the quantum circuit shown in Fig.~\ref{fig:sud_invariant_channel}~(a) using the quantum channel $\widetilde{\Psi}: \bigoplus_{\lambda\vdash_d n} \End(\CC^{m_{\lambda}}) \to \bigoplus_{\mu \vdash_d n} \End(\CC^{m_{\mu}})$ defined by $\widetilde{\Psi} \defeq \bigoplus_{\lambda, \mu \vdash_d n} \Psi^{\lambda\mu}$ and the quantum Schur transform.

\begin{figure}
    \centering
    \includegraphics[width=\linewidth]{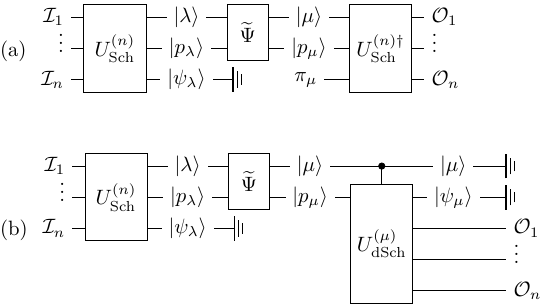}
    \caption{(a) The $\U(d)$-invariant channel can be implemented using the quantum Schur transforms and a quantum channel $\widetilde{\Psi}$ on the multiplicity space and the preparation of the maximally mixed state $\pi_\mu$.
    The ground symbol in the circuit represents the partial trace.
    This implementation cannot be streamed due to the presence of $\pi_\mu$.\\
    (b) The $\U(d)$-invariant channel can also be implemented using the quantum Schur transform, the dual Schur transform and a quantum channel $\widetilde{\Psi}$ on the multiplicity space.
    This implementation can be streamed as shown in the SM~\cite{supple}.}
    \label{fig:sud_invariant_channel}
\end{figure}

However, this implementation does not work for the quantum comb (or non-signalling channel).
Though the quantum Schur transform can be implemented in a streaming manner using the Clebsch--Gordan (CG) transforms as shown in Refs.~\cite{harrow2005applications,bacon2006efficient,bacon2007quantum}, the preparation of the maximally mixed state $\pi_{\mu}$ in the middle of the circuit should be done after the quantum channel $\widetilde{\Psi}$, which prohibits the overall circuit to be done in a streaming manner.

\subsection{Implementation of quantum comb with group invariance}
\label{subsec:canonical_decomposition}
We present the main result of this work, which shows the implementation of the quantum comb with $(G\times H)$-invariance for any compact groups $G$ and $H$.
We consider the ($G \times H$)-invariance on the Choi matrix of the quantum comb given by
\begin{multline}
    \label{eq:gh_symmetry}
    \sof[\big]{C, \of[\big]{\otimes_{i=1}^{n}\overline{\rho}_i(g)}_{\I^n} \otimes \of[\big]{\otimes_{j=1}^{n} \sigma_j(h)}_{\mcO^n}} = 0 \\ \quad \forall (g, h) \in G \times H,
\end{multline}
where $\rho_i$ and $\sigma_j$ are unitary representations of the compact groups $G$ and $H$ with the representation spaces $\mcI_i$ and $\mcO_j$ for $i, j\in [n]$, respectively, and $\overline{\rho}_i$ is the complex conjugate representation of $\rho_i$.
The ($G \times H$)-invariance~\eqref{eq:gh_symmetry} corresponds to the ($G \times H$)-invariance of the corresponding quantum channel $\Psi$ given by
\begin{align}
    \Psi\sof[\big]{\rho(g) \cdot \rho(g)^\dagger} = \sigma(h) \Psi(\cdot) \sigma(h)^\dagger
    \quad \forall (g, h) \in G \times H,
\end{align}
where $\rho(g) \defeq \bigotimes_{i=1}^{n} \rho_i(g)$ and $\sigma(h) \defeq \bigotimes_{j=1}^{n} \sigma_j(h)$.

To state the main result, we introduce the CG transforms corresponding to the representations $\rho_i$ of a compact group $G$.
Due to the Peter--Weyl theorem, the tensor product of representations $\bigotimes_{i=1}^{n} \rho_i$ is decomposed into the irreps as
\begin{align}
    \bigotimes_{i=1}^{n} V_{\rho_i} \cong \bigoplus_{\lambda\in \Irr{G}^{(n)}} V_{\lambda} \otimes \CC^{M_{\lambda}^{(n)}},
\end{align}
where $\Irr{G}^{(n)}$ is the set of irreps appearing in the tensor product $\bigotimes_{i=1}^{n} \rho_i$, $V_{\lambda}$ is the representation space of an irrep $\lambda$, and $M_{\lambda}^{(n)}$ is the multiplicity of an irrep $\lambda$ in the tensor product $\bigotimes_{i=1}^{n} \rho_i$.
This decomposition is obtained recursively by considering the following irrep decomposition of $\lambda_{i-1} \otimes \rho_i$ for $\lambda_{i-1}\in \Irr{G}^{(i-1)}$:
\begin{align}
    \label{eq:cg_isomorphism}
    V_{\lambda_{i-1}} \otimes V_{\rho_i} \cong \bigoplus_{\lambda_i \in \Irr{G}^{(i)}} V_{\lambda_i} \otimes \CC^{C_{\lambda_{i-1} \rho_i}^{\lambda_i}},
\end{align}
where $C_{\lambda_{i-1} \rho_i}^{\lambda_i}$ is the multiplicity of an irrep $\lambda_i$ in the tensor product $\lambda_{i-1} \otimes \rho_i$.
The CG transform is defined as the isomorphism corresponding to Eq.~\eqref{eq:cg_isomorphism}.
In the quantum circuit, it is represented as an isometry operator
\begin{align}
    \mathrm{CG}_{\rho_i}: \mcR_{i-1}^{G} \otimes \mcV_{i-1}^G \otimes V_{\rho_i} \to \mcR_{i-1}^{G} \otimes \mcC_{i}^{G} \otimes \mcR_i^G \otimes \mcV_{i}^{G},
\end{align}
where $\mcR_i^G$ stores the label of irreps $\lambda_i\in \Irr{G}^{(i)}$, $\mcV_i^G$ stores the vectors $\ket{\psi_{\lambda_i}} \in V_{\lambda_i}$ for irreps $\lambda_i\in \Irr{G}^{(i)}$, and $\mcC_{i}^{G}$ stores the vectors $\ket{a_i}\in \CC^{C_{\lambda_{i-1} \rho_i}^{\lambda_i}}$ for irreps $\lambda_{i-1}\in \Irr{G}^{(i-1)}$ and $\lambda_i \in \Irr{G}^{(i)}$.

Using the CG transforms, the quantum comb with $G\times H$ symmetry is implemented as shown in the following theorem.

\begin{figure*}
    \centering
    \includegraphics{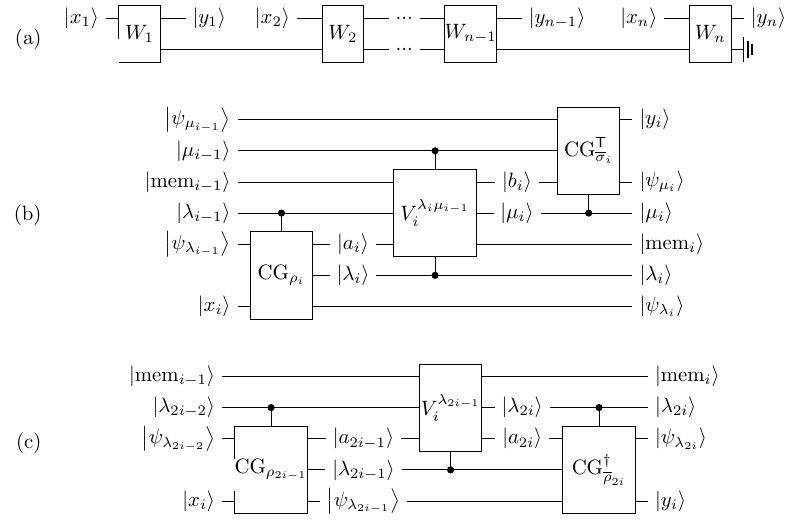}
    \caption{(a) Implementation of the quantum comb with the $(G\times H)$-invariance~\eqref{eq:gh_symmetry}, which is composed of isometries $W_i$ given in (b).
    The ground symbol in the circuit represents the partial trace.\\
    (b) The isometry $W_i$ for the $(G\times H)$-invariant quantum comb is given by the CG transforms and an arbitrary isometry operator $V_i^{\lambda_i \mu_{i-1}}$ shown in Eq.~\eqref{eq:V_i}.\\
    (c) The isometry $W_i$ for the $G$-covariant quantum comb is given by the CG transforms and an arbitrary isometry operator $V_i^{\lambda_{2i-1}}$.}
    \label{fig:comb_implementation}
\end{figure*}

\begin{theorem}
    \label{thm:comb_implementation}
    Any quantum comb with its Choi matrix $C\in \End(\mcI^n\otimes \mcO^n)$ satisfying Eq.~\eqref{eq:gh_symmetry} can be realized in the form of Fig.~\ref{fig:comb_implementation}~(a) with isometry $W_i$ defined in Fig.~\ref{fig:comb_implementation}~(b), where $V_i^{\lambda_i \mu_{i-1}}$ is an isometry operator corresponding to $\lambda_i\in \Irr{G}^{(i)}$ and $\mu_{i-1}\in\Irr{H}^{(i-1)}$, whose input and output spaces are given by
    \begin{align}
        \label{eq:V_i}
        V_i^{\lambda_i \mu_{i-1}}: \mcM_{i-1} \otimes \mcR_{i-1}^G \otimes \mcC_i^G \to \mcC_i^H \otimes \mcR_i^H \otimes \mcM_i,
    \end{align}
    where $\mcM_i$ is an auxiliary space storing the vector $\ket{\mathrm{mem}_i}\in \mcM_i$.
\end{theorem}
\begin{proof}[Proof sketch]
    We provide the main idea of the circuit construction for the case of $G = H = \U(d)$ and $\rho_i(g) = \sigma_j(g) = g$, and leave the full proof to Sec.~\ref{sec:proof_comb_implementation}.
    The difficulty of constructing the streaming circuit comes from the preparation of the maximally mixed state $\pi_{\mu}$, which is required to apply the inverse Schur transform $U_\mathrm{Sch}^{(n)\dagger}$.
    
    To circumvent this problem, we define the \emph{dual Schur transform}, which does not require the irrep space register as an input.
    Suppose $\{\ket{\psi_\lambda^{(i)}}\}$ is an orthonormal basis of the irrep space $V_{\lambda}$ and $\overline{U}_{\lambda}$ is the complex conjugate of $U_\lambda$ in this basis.
    Then, the maximally entangled state defined by $\ket{\phi^+_\lambda}\defeq {1\over \sqrt{d_\lambda}}\sum_{i\in [d_\lambda]} \ket{\psi_\lambda^{(i)}}\otimes \ket{\psi_\lambda^{(i)}}$ for $d_\lambda\coloneqq \dim V_\lambda$ is invariant under the tensor product $U_\lambda \otimes \overline{U}_{\lambda}$, i.e., the tensor product representation $U_\lambda \otimes \overline{U}_{\lambda}$ contains the trivial representation $\varnothing$.
    From Eq.~\eqref{eq:sw_duality_space}, we have
    \begin{align}
        (\CC^d)^{\otimes n} \otimes V_\lambda &\cong V_\lambda \otimes \bigoplus_{\mu \vdash_d n}  V_{\overline{\mu}} \otimes \CC^{m_{\mu}},
    \end{align}
    where $V_{\overline{\mu}}$ is the irrep space of $\overline{U}_{\mu}$.
    Since $V_\lambda \otimes V_{\overline{\lambda}}$ contains the trivial representation space $V_\varnothing\simeq \CC$, $(\CC^d)^{\otimes n} \otimes V_{\lambda}$ contains a subspace isomorphic to $\CC^{m_{\lambda}}$.
    Therefore, we can define the isometry
    \begin{align}
        U_{\mathrm{dSch}}^{(\lambda)}: \CC^{m_\lambda} \to (\CC^d)^{\otimes n} \otimes V_{\lambda},
    \end{align}
    and we call $U_{\mathrm{dSch}}^{(n)} \defeq \oplus_{\lambda \in \hat{G}^n} U_{\mathrm{dSch}}^{(\lambda)}$ the dual Schur transform.
    
    Using the dual Schur transform, the quantum channel $\Psi$ can be written as
    \begin{align}
        \Psi(\cdot) = \bigoplus_{\lambda, \mu \vdash_d n} \Tr_{V_\mu}\left[U_{\mathrm{dSch}}^{(\mu)} \Psi^{\lambda\mu} \circ \Tr_{V_\lambda} \left( \Pi_\lambda \cdot \right) U_{\mathrm{dSch}}^{(\mu) \dagger}\right],
    \end{align}
    which can be implemented by the quantum circuit shown in Fig.~\ref{fig:sud_invariant_channel}~(b) using the quantum channel $\widetilde{\Psi} \defeq \bigoplus_{\lambda, \mu \vdash_d n} \Psi^{\lambda\mu}$, the quantum Schur transform, and the dual Schur transform.
    When the Choi matrix $C$ satisfies the comb condition~\eqref{eq:comb_condition}, the quantum channel $\widetilde{\Psi}$ is implementable in a streaming manner.
    The quantum Schur transform is implemented by a sequence of CG transforms.
    As shown in the SM, the dual Schur transform is implemented by using a sequence of the dual CG transforms.
    By putting it all together, we obtain the implementation of the quantum comb shown in Fig.~\ref{fig:comb_implementation}.
\end{proof}

\subsection{Extension to $G$-covariant quantum comb}
\label{subsec:covariant_comb}
We consider the $G$-covariant quantum comb for any compact group $G$, which satisfies
\begin{align}
    \label{eq:g_covariance}
    \left[C, \of[\big]{\otimes_{i=1}^{n}\overline{\rho}_{2i-1}(g)}_{\I^n} \otimes \of[\big]{\otimes_{j=1}^{n} \overline{\rho}_{2j}(g)}_{\mcO^n}\right] = 0 \quad \forall g \in G,
\end{align}
where $\rho_i$ for $i\in [2n]$ are unitary representations of the compact group $G$.
The $G$-covariance~\eqref{eq:g_covariance} corresponds to the $G$-covariance of the corresponding quantum channel $\Psi$ given by
\begin{align}
    \Psi\sof[\big]{\rho(g) \cdot \rho(g)^\dagger} = \overline{\sigma}(g) \Psi(\cdot) \overline{\sigma}(g)^\dagger \quad \forall g\in G,
\end{align}
where $\rho(g) \defeq \bigotimes_{i=1}^{n} \rho_{2i-1}(g)$ and $\overline{\sigma}(g) \defeq \bigotimes_{j=1}^{n} \overline{\rho}_{2j}(g)$.
Since the $(G\times H)$-covariance with the representations
\begin{align}
    \begin{cases}
        \rho_{2i-1}'(g,h) &\defeq \rho_{i}(g),\\
        \rho'_{2i}(g,h) &\defeq \overline{\sigma}_{i}(h),
    \end{cases}
    \quad  \forall (g,h)\in G\times H,
\end{align}
is equivalent to the $(G\times H)$-invariance~\eqref{eq:gh_symmetry}, the $G$-covariant quantum comb can be considered as a generalization of the $(G\times H)$-invariant quantum comb.
The $G$-covariant quantum comb is also implemented in a streaming manner using the CG transforms and isometry operators on the multiplicity spaces, as shown below.

\begin{theorem}
    \label{thm:comb_implementation_covariant}
    Any quantum comb with its Choi matrix $C\in \End(\mcI^n\otimes \mcO^n)$ satisfying the $G$-covariance~\eqref{eq:g_covariance} can be realized in the form of Fig.~\ref{fig:comb_implementation}~(a) with isometry $W_i$ defined in Fig.~\ref{fig:comb_implementation}~(c), where $V_i^{\lambda_{2i-1}}$ is an isometry operator corresponding to $\lambda_{2i-1}\in \Irr{G}^{(2i-1)}$, whose input and output spaces are given by
    \begin{align}
        V_i^{\lambda_{2i-1}}: \mcM_{i-1} \otimes \mcR_{2i-1}^G \otimes \mcC_{2i-1}^G \to \mcM_{i} \otimes \mcR_{2i}^G \otimes \mcC_{2i}^G,
    \end{align}
    where $\mcM_{i}, \mcR_{i}^G, \mcC_{i}^G$ are defined similarly to Thm.~\ref{thm:comb_implementation}.
\end{theorem}
\begin{proof}
See Sec.~\ref{sec:proof_comb_implementation_covariant} for the proof.
\end{proof}

\section{Application: Parametrized quantum comb with group symmetry}
\label{sec:application}
Theorem~\ref{thm:comb_implementation} provides a canonical form of the quantum comb with $G\times H$ symmetry, which can be used to parametrize the $(G\times H)$-invariant quantum comb.
As an application, we consider the tasks of unitary inversion and unitary transposition, whose optimal protocols have the $\U(d)\times \U(d)$-invariance ($G=H=\U(d)$ for these tasks).
We assume that the dimension of the auxiliary space $\mcM_i$ is one-dimensional, and optimize the isometry operators $V_i^{\lambda_i \mu_{i-1}}$ in Eq.~\eqref{eq:V_i} to maximize the average fidelity of the output channel~\eqref{eq:comb_output} with the target channel $f(\mathcal{U}) = \mcU^{-1}$ or $f(\mcU) = \mcU\tp$ for $U\in \U(d)$, where $\mcU(\cdot) \defeq U\cdot U^\dagger$, $\mcU^{-1}(\cdot) \defeq U^{-1} \cdot U$, and $\mcU\tp(\cdot) \defeq U\tp \cdot U^*$.
The average fidelity is given by
\begin{align}
    F \defeq \int \dd U \, F_\mathrm{ch}\sof[\big]{\mathcal{C}(\mathcal{U}^{\times n}), f(\mcU)},
\end{align}
where $F_\mathrm{ch}$ is the channel fidelity and $\dd U$ is the Haar measure of $\U(d)$.
Then, we formulate a non-linear optimization problem for optimizing the comb fidelity following the methodology of \cite{mo2025parameterized}.

\begin{table*}
\begin{tabular}{c|c||c|c|c|c|c|c|c}
    $f(U)$   &   \diagbox{$d$}{$n$}     &   1       &  2        &  3       &  4       &  5        & 6           & 7           \\ \hline \hline
    \multirow{6}{*}{$U\tp$}     &   2   &\g  0.5    &\g   0.75  &\g0.933013&\g 1      & -         & -           & -           \\
                                &   3   &\g0.2222222&\g 0.407408&\g0.626597&\g0.799251&\g0.932387 & 0.998243    & 1.000000(0) \\
                                &   4   &\g  0.125  &\g 0.21875 &\g0.362903&\g0.543985&\g0.697606 & 0.828739(9) & ?           \\
                                &   5   &\g  0.08   &\g  0.136  &\g0.214954&\g0.331873&\g0.482961 &0.616832(1)  & ?           \\
                                &   6   &\g0.0555556&\g0.0925926&\g0.141901&\g0.209426&0.307098(3)&0.341168(7)  & ?           \\
                                &   7   &\g0.0408163&\g0.0670554&\g 0.10059&\g0.144193&0.202858(7)&0.286130(4)  & ?           \\ \hline
    \multirow{6}{*}{$U\ct$}     &   2   &\g   0.5   &\g   0.75  &\g0.933013&\g 1      & -         & -           & -           \\
                                &   3   &\g 0.222222&\g 0.333333&\g0.444445&\g0.555556&\g0.666667 &\g0.777777(7)&0.887137(1)  \\
                                &   4   &\g   0.125 &\g  0.1875 &\g  0.25  &\g0.3125  &\g0.375    &\g0.437500(0)&\g0.500000(0)\\
                                &   5   &\g   0.08  &\g  0.12   &\g  0.16  &\g 0.2    &\g0.24     &\g0.280000(0)&\g0.320000(0)\\
                                &   6   &\g0.0555556&\g0.0833334&\g0.111111&\g0.138889&\g0.166667 &\g0.194444(4)& ?           \\
                                &   7   &\g0.0408163&\g0.0612245&\g0.081632&\g0.102041&\g0.122449 &\g0.142857(1)& ?            
\end{tabular}
\caption{\textbf{Lower bounds from nonlinear optimization}. Fidelity according to nonlinear optimization in Julia.
The green values are the ones that match with the SDP and analytical results up to the fourth digit. These numbers should be treated as lower bounds for the optimal values.}
\label{table:nlopt}
\end{table*}

\begin{table*}
\begin{tabular}{c|c||c|c|c|c|c|c|c}
    $f(U)$   &   \diagbox{$d$}{$n$}     &   1                      &  2                      &  3                      &  4                     &  5                     &  6                          & 7                          \\ \hline \hline
    \multirow{6}{*}{$U\tp$}     &   2   &  \cellblue 0.5           & 0.749999(9)             & 0.933012(5)             & \cellblue 1            & -                      &  -                          & -                          \\
                                &   3   &  \cellblue 0.222222...   & 0.407407(3)             & 0.626596(4)             & 0.799250(3)            & 0.932375(8)            &  \cellred 0.998243          &  ?                         \\
                                &   4   &  \cellblue 0.125         & 0.218749(9)             & 0.362903(1)             & 0.544148(0)            & 0.697604(2)            &  ?                          &  ?                         \\
                                &   5   &  \cellblue 0.08          & 0.135999(9)             & 0.214953(5)             & 0.331870(6)            & 0.482926(8)            &  ?                          &  ?                         \\
                                &   6   &  \cellblue 0.0555555...  & 0.092592(5)             & 0.141901(1)             & 0.209438(5)            & ?                      &  ?                          &  ?                         \\
                                &   7   &  \cellblue 0.0408163...  & 0.067055(2)             & 0.100590(1)             & 0.144190(3)            & ?                      &  ?                          &  ?                         \\ \hline
    \multirow{6}{*}{$U\ct$}     &   2   &  \cellblue$0.5$          & 0.749999(9)             & 0.933012(5)             & \cellblue$1$           & -                      &  -                          &  -                         \\
                                &   3   &  \cellblue$0.222222...$  & \cellblue$0.333333...$  & 0.444444(3)             & 0.555555(5)            & 0.666666(4)            & \cellblue$\leq 0.777777...$ & \cellblue$\leq 0.888888...$\\
                                &   4   &  \cellblue$0.125$        & \cellblue$0.1875$       & \cellblue$0.25$         & $0.312499(9)$          & $0.374999(9)$          & \cellblue$\leq 0.4375$      & \cellblue$\leq 0.5$        \\
                                &   5   &  \cellblue$0.08$         & \cellblue$0.12$         & \cellblue$0.16$         & \cellblue$0.2$         & $0.240000(0)$          & \cellblue$\leq 0.28$        & \cellblue$\leq 0.32$       \\
                                &   6   &  \cellblue$0.0555555...$ & \cellblue$0.0833333...$ & \cellblue$0.111111...$  & \cellblue$0.138888...$ & \cellblue$0.166666...$ & \cellblue$\leq 0.194444...$ & \cellblue$\leq 0.222222...$\\
                                &   7   &  \cellblue$0.0408163...$ & \cellblue$0.0612244...$ & \cellblue$0.0816326...$ & \cellblue$0.102040...$ & \cellblue$0.122448...$ & \cellblue$0.142857...$      & \cellblue$\leq 0.163265...$\\
\end{tabular}
\caption{\textbf{SDP data and known analytical results}. The optimal fidelity of the universal \textbf{sequential} unitary transposition and inversion. The numerical calculation is done in Julia using MOSEK optimiser.
The blue numbers correspond to the analytical results given by $F_\mathrm{trans} = F_\mathrm{inv} = {2\over d^2}$ for $n=1$~\cite{chiribella2016optimal}, $F_\mathrm{trans}= F_\mathrm{inv}=1$ for $d=2, n=4$~\cite{yoshida2023reversing}, $F_\mathrm{inv} = {n+1\over d^2}$ for $n\leq d-1$~\cite{yoshida2024one}, and $F_\mathrm{inv}\leq {n+1\over d^2}$ for $n\geq d$~\cite{chen2025tight}.
The red value corresponds to the heuristic SDP result, which only gives a lower bound for the optimal value.
Notation $0.xxxxxx(y)$ means that the number has the form $0.xxxxxxy...$, i.e., the digit $y$ is not rounded.}
\label{table:sdp}
\end{table*}

We solve this optimization problem numerically using Julia with JuMP with Ipopt and Gurobi optimizers to solve the resulting nonlinear optimization problems. 
Our results are summarised in \cref{table:nlopt}.
We obtain the optimal average fidelities for unitary inversion and transposition for $d=3$ up to $n=7$ as shown in the SM~\cite{supple}, and find the exact unitary transposition protocol with $n=7$ queries.
This result is improved over the previous best result $n=13$ in Ref.~\cite{chen2024quantum}.
We also justify the assumption of $\dim \mcM_i = 1$ by comparing the obtained fidelities with that in the SDP approach~\cite{quintino2022deterministic, grinko2024linear} for $n\leq 5$ and analytical results~\cite{chiribella2008qca,quintino2022deterministic,yoshida2023reversing,yoshida2024one,chen2025tight} shown in \cref{table:sdp}.
The analytical results for the optimal fidelity of unitary transposition, denoted by $F_\mathrm{trans}$, and that of unitary inversion, denoted by $F_\mathrm{inv}$ are given as follows:
\begin{itemize}
    \item \cite{chiribella2016optimal} $F_\mathrm{trans} = F_\mathrm{inv} = {2\over d^2}$ for $n=1$,
    \item \cite{yoshida2023reversing} $F_\mathrm{trans} = F_\mathrm{inv} = 1$ for $d=2, n=4$,
    \item \cite{yoshida2024one} $F_\mathrm{inv} = {n+1\over d^2}$ for $n\leq d-1$,
    \item \cite{chen2025tight} $F_\mathrm{inv} \leq {n+1\over d^2}$ for $n\geq d$.
\end{itemize}

We also compare the number of variables in our approach with that in the naive approach~\cite{mo2025parameterized} to reproduce the same circuit.
Our approach significantly reduces the number of variables, which enables us to perform the optimization for a larger query number $n$, as shown in Appendix~\ref{sec:variable_count}.

\section{Conclusion and outlook}
\label{sec:conclusion}
We show that any $(G\times H)$-invariant quantum comb for any compact groups $G$ and $H$ can be implemented using the corresponding CG transforms and isometry operators on the multiplicity spaces, which is extended to the $G$-covariant quantum comb.
We show an application of this result to the parametrized quantum comb, which significantly reduces the number of variables in the optimization.
From this optimization, we obtain an exact qutrit-unitary transposition protocol with $n=7$ queries, which improves the previous best result $n=13$ in Ref.~\cite{chen2024quantum}.
Another application for the simulation of random unitary is shown in the concurrent work~\cite{grinko2025quantum} of some of the authors, and a similar circuit for the case of $G = H = \U(d)$ appears in the context of transforming isometry channels~\cite{yoshida2025universal}.

This work opens several future directions.
Given the ubiquitous nature of the symmetry in quantum physics~\cite{noether1971invariant, itzykson1966unitary} and the quantum comb in various settings such as the dynamical resource theory~\cite{chitambar2019quantum}, higher-order quantum transformation~\cite{taranto2025higher}, non-Markovian dynamics~\cite{pollock2018non}, and causal structure~\cite{oreshkov2012quantum}, we expect that our results will be widely used in the analysis and implementation of quantum information processing.
We also expect that the $G$-covariant quantum comb finds its applications in various settings, such as the streaming implementation of the $G$-covariant quantum channel and simulation of $G$-covariant dynamics of open quantum systems.
Our application of the canonical form of the quantum comb to the parametrized quantum comb can be considered as a truncation of the search space of the optimization by using the Schur basis.
It would be interesting to investigate the performance of this approach for other tasks and other groups, and to find a systematic way to choose the truncation of the search space.

\acknowledgments
This work was supported by the MEXT
Quantum Leap Flagship Program (MEXT QLEAP) JPMXS0118069605, JPMXS0120351339, Japan Society for the Promotion of Science (JSPS) KAKENHI Grants No. 23KJ0734 and No. 23K21643, FoPM, WINGS Program, the University of Tokyo, DAIKIN Fellowship Program, the University of Tokyo, IBM Quantum, NWO grant NGF.1623.23.025 (“Qudits in theory and experiment”) and NWO Vidi grant (Project No. VI.Vidi.192.109).

\let\oldaddcontentsline\addcontentsline
\renewcommand{\addcontentsline}[3]{}
\bibliography{main}
\let\addcontentsline\oldaddcontentsline

\clearpage
\onecolumngrid
\renewenvironment{widetext}{}{}

\appendix

\section{Proof of main results}
\label{sec:proof_main_results}

This appendix provides the proof of the main results.
We express the comb condition in the commutant of the tensor representation in Sec.~\ref{sec:comb_condition_symmetry}, and introduce the generalized Schur and dual Schur transforms corresponding to the tensor representation in Sec.~\ref{sec:generalized_schur}.
Using these ingredients, we construct the quantum circuit for the quantum comb with $G\times H$ symmetry and that for the $G$-covariant quantum comb, which correspond to the proofs of
Theorem~\ref{thm:comb_implementation} and Theorem~\ref{thm:comb_implementation_covariant} in Sec.~\ref{sec:proof_comb_implementation} and Sec.~\ref{sec:proof_comb_implementation_covariant}, respectively.

\subsection{Comb condition in the commutant of the tensor representation}
\label{sec:comb_condition_symmetry}

As shown in Sec.~\ref{sec:quantum_comb}, the quantum comb is characterized as a Choi matrix $C\in \End(\mcI^n \otimes \mcO^n)$ satisfying the comb condition given by Eq.~\eqref{eq:comb_condition}.
where $\mcI^n$ and $\mcO^n$ are joint Hilbert spaces given by $\mcI^n \defeq \bigotimes_{i=1}^n \mcI_i$ and $\mcO^n \defeq \bigotimes_{i=1}^n \mcO_i$, respectively, $C_{n} \defeq C$, and $C_{i-1} \defeq \frac{1}{d}\Tr_{\mcI_i \mcO_i} C_i$.
Suppose the Choi operator $C$ of a quantum comb satisfies the $(G,H)-$invariance given by Eq.~\eqref{eq:gh_symmetry}.
Then, $C$ belongs to the commutant of the tensor representation $\bigotimes_{i=1}^{n}\rho_i \otimes \bigotimes_{j=1}^{n} \sigma_j$, i.e.,
\begin{align}
    C &\in \mathrm{Comm}\left(\left(\otimes_{i=1}^{n}\overline{\rho}_i\right) \otimes \left(\otimes_{j=1}^{n} \sigma_j\right)\right)
    = \mathrm{span}\left\{\overline{E}^{\lambda}_{pq} \otimes \widetilde{E}^{\mu}_{rs} \; \middle| \;
    \begin{aligned}
        \lambda\in \Irr{G}^{(n)}, \mu\in \Irr{H}^{(n+1)},\\
        p, q\in \Paths(\lambda, \scB_L),\\
        r,s\in \Paths(\mu, \scB_R)
    \end{aligned} \right\},
\end{align}
where $E^{\lambda}_{pq}$ and $\widetilde{E}^\mu_{rs}$ are the orthogonal basis of the commutant of the tensor representation $\bigotimes_{i=1}^{n} \rho_i$ and $\bigotimes_{j=1}^{n} \sigma_j$, and $\scB_L$ and $\scB_R$ are the corresponding Bratteli diagrams, respectively, defined similarly to Eq.~\eqref{eq:def_E}.
Therefore, $C$ is written associated with a set of matrices $\{C^{\lambda\mu}\in \End[\CC^{M_\lambda^{(n)}} \otimes \CC^{M_\mu^{(n)}}] \mid \lambda\in \Irr{G}^{(n)}, \mu\in \Irr{H}^{(n)}\}$ as
\begin{align}
    \label{eq:decomposition_C}
    C = \sum_{\substack{\lambda\in \Irr{G}^{(n)}\\\mu\in \Irr{H}^{(n)}}} \sum_{\substack{p, q\in \Paths(\lambda, \scB_L)\\r, s\in \Paths(\mu, \scB_R)}} C^{\lambda\mu}_{pr, qs} {\overline{E}^{\lambda}_{pq} \over d_\lambda} \otimes {\widetilde{E}^{\mu}_{rs} \over d_\mu},
\end{align}
where $C^{\lambda\mu}_{pr, qs}$ are the matrix elements of $C^{\lambda\mu}$ given by
\begin{align}
    C^{\lambda\mu}_{pr, qs} \defeq \bra{q}\bra{s} C^{\lambda\mu} \ket{p}\ket{r}
\end{align}
for $\lambda\in\Irr{G}^{(n)}$, $\mu\in \Irr{H}^{(n)}$, $p, q\in \Paths(\lambda, \scB_L)$, and $r, s \in \Paths(\mu, \scB_R)$.
The comb condition~\eqref{eq:comb_condition} is translated into the following conditions on $\{C^{\lambda\mu}\}$.
\begin{widetext}
\begin{lemma}
    \label{lem:comb_condition_symmetry}
    The operator $C$ in the form of Eq.~\eqref{eq:decomposition_C} satisfies the comb condition~\eqref{eq:comb_condition} if and only if the set of matrices $\{C^{\lambda\mu}\in \End[\CC^{M_\lambda^{(n)}} \otimes \CC^{M_\mu^{(n)}}] \mid \lambda\in \Irr{G}^{(n)}, \mu\in \Irr{H}^{(n)}\}$ satisfies the following conditions:
    \begin{align}
        C^{\lambda\mu} &\succeq 0 \qquad \forall \lambda\in \Irr{G}^{(n)}, \forall \mu\in \Irr{H}^{(n)},\\
        \begin{split}
        \sum_{\mu_i\in \Irr{H}^{(i)}} & \sum_{a\in [C_{\mu_{i-1}}^{\mu_i}]} (\1_{M_{\lambda_i}^{(i)}}\otimes R_a^{\mu_i\to \mu_{i-1}}) C_{i}^{\lambda_i \mu_i} (\1_{M_{\lambda_i}^{(i)}}\otimes R_a^{\mu\to \beta})^\dagger = \\
        &= \sum_{\lambda_{i-1}\in \Irr{G}^{(i-1)}}\sum_{b\in [C_{\lambda_{i-1}}^{\lambda_i}]} {d_{\lambda_i} \over d_{\lambda_{i-1}}} (A_b^{\lambda_{i-1}\to\lambda_i}\otimes \1_{M_{\mu_{i-1}}^{(i-1)}}) C_{i-1}^{\lambda_{i-1} \mu_{i-1}} (A_b^{\lambda_{i-1}\to\lambda_i}\otimes \1_{M_{\mu_{i-1}}^{(i-1)}})^\dagger \\ & \hspace{+7cm}  \forall i\in [n], \forall \lambda_i\in \Irr{G}^{(i)}, \forall \mu_{i-1}\in \Irr{H}^{(i-1)},
        \end{split}\\
        C_0^{\varnothing\varnothing} &= 1,
    \end{align}
    where $R_a^{\mu_i\to \mu_{i-1}}: \CC^{M_{\mu_i}^{(i)}} \to \CC^{M_{\mu_{i-1}}^{(i-1)}}$ and $A_b^{\lambda_{i-1}\to\lambda_i}: \CC^{M_{\lambda_{i-1}}^{(i-1)}} \to \CC^{M_{\lambda_i}^{(i)}}$ are defined by
    \begin{align}
        R_a^{\mu_i\to \mu_{i-1}} &\defeq \sum_{p\in \Paths(\mu_{i-1}, \scB_R)} \ketbra{p}{p\xrightarrow{a}\mu_i},\\
        A_b^{\lambda_{i-1}\to\lambda_i} &\defeq \sum_{p\in \Paths(\lambda_{i-1}, \scB_L)} \ketbra{p\xrightarrow{b} \lambda_i}{p},
    \end{align}
    which corresponds to removing and adding an edge to a path, and $C_i^{\lambda \mu}$ is recursively defined as
    \begin{align}
        C_{n}^{\lambda \mu} &\defeq C^{\lambda \mu},\\
        C_{i-1}^{\lambda_{i-1}\mu_{i-1}} &\defeq {1\over d_{\mcI_i}}\sum_{\lambda_i\in \Irr{G}^{(i)}} \sum_{\mu_i\in \Irr{H}^{(i)}} \sum_{a\in [C_{\mu_{i-1}}^{\mu_{i}}]} \sum_{b\in [C_{\lambda_{i-1}}^{\lambda_i}]} (R_b^{\lambda_i \to \lambda_{i-1}}\otimes R_a^{\mu_i\to\mu_{i-1}}) C_i^{\lambda_i \mu_i} (R_b^{\lambda_i \to \lambda_{i-1}}\otimes R_a^{\mu_i\to\mu_{i-1}})^\dagger,
    \end{align}
    for $i\in [n]$, and $d_{\mcI_i} \coloneqq \dim \mcI_i$.
\end{lemma}
\end{widetext}
\begin{proof}
    This lemma directly follows from Lems.~\ref{lem:orthogonal_basis_properties} and \ref{lem:tensor_and_partial_trace}, similarly to Thm.~S7 in Ref.~\cite{yoshida2023reversing}.
\end{proof}

Similarly, we consider the $G$-covariant quantum comb, which satisfies Eq.~\eqref{eq:g_covariance}.
Then, $C$ belongs to the commutant of the tensor representation $\bigotimes_{i=1}^{2n}\overline{\rho}_{i}$, i.e.,
\begin{align}
    C &\in \mathrm{Comm}\left(\otimes_{i=1}^{2n}\overline{\rho}_i\right) 
    = \mathrm{span}\left\{\overline{E}^{\lambda}_{pq} \; \middle| \;\lambda\in \Irr{G}^{(2n)}, p,q\in \Paths(\lambda, \scB) \right\},
\end{align}
where $\scB$ is the Bratteli diagram corresponding to the tensor representation $\otimes_{i=1}^{2n} \rho_i$.
Therefore, $C$ is written associated with a set of matrices $\{C^{\lambda}\in \End[\CC^{M_\lambda^{(2n)}}] \mid \lambda\in \Irr{G}^{(2n)}\}$ as
\begin{align}
    \label{eq:decomposition_C_G}
    C = \sum_{\lambda\in \Irr{G}^{(2n)}} \sum_{p, q\in \Paths(\lambda, \scB)} C^{\lambda}_{pq} {\overline{E}^{\lambda}_{pq} \over d_\lambda},
\end{align}
where $C^{\lambda}_{pq}$ are the matrix elements of $C^{\lambda}$ given by
\begin{align}
    C^{\lambda}_{pq} = \bra{q} C^{\lambda} \ket{p} \quad \forall p, q\in \Paths(\lambda, \scB).
\end{align}
The comb condition~\eqref{eq:comb_condition} is translated into the following conditions on $\{C^{\lambda}\}$.
\begin{widetext}
\begin{lemma}
    \label{lem:comb_condition_covariance}
    The operator $C$ in the form of Eq.~\eqref{eq:decomposition_C_G} satisfies the comb condition~\eqref{eq:comb_condition} if and only if the set of matrices $\{C^{\lambda}\in \End[\CC^{M_\lambda^{(2n)}}] \mid \lambda\in \Irr{G}^{(2n)}\}$ satisfies the following conditions:
    \begin{align}
        C^{\lambda} &\succeq 0 \quad \forall \lambda\in \Irr{G}^{(2n)},\\
        \begin{split}
        \sum_{\lambda_{2i}\in \Irr{G}^{(2i)}} & \sum_{a\in [C_{\lambda_{2i-1}}^{\lambda_{2i}}]} (R_a^{\lambda_{2i}\to \lambda_{2i-1}}) C_{i}^{\lambda_{2i}} (\1_{M_{\lambda_{2i}}^{(2i)}}\otimes R_a^{\lambda_{2i}\to \lambda_{2i-1}})^\dagger\\
        &= \sum_{\lambda_{2i-2}\in \Irr{G}^{(2i-2)}}\sum_{b\in [C_{\lambda_{2i-2}}^{\lambda_{2i-1}}]} {d_{\lambda_{2i-1}} \over d_{\lambda_{2i-2}}} (A_b^{\lambda_{2i-2}\to\lambda_{2i-1}}) C_{i-1}^{\lambda_{2i-2}} (A_b^{\lambda_{2i-2}\to\lambda_{2i-1}})^\dagger \quad \forall i\in [n], \forall \lambda_{2i-1}\in \Irr{G}^{(2i-1)},
        \end{split}\\
        C_0^{\varnothing} &= 1,
    \end{align}
    where $R_a^{\lambda_{2i}\to \lambda_{2i-1}}: \CC^{M_{\lambda_{2i}}^{(2i)}} \to \CC^{M_{\lambda_{2i-1}}^{(2i-1)}}$ and $A_b^{\lambda_{2i-2}\to\lambda_{2i-1}}: \CC^{M_{\lambda_{2i-2}}^{(2i-2)}} \to \CC^{M_{\lambda_{2i-1}}^{(2i-1)}}$ are defined in Lem.~\ref{lem:comb_condition_symmetry}, and $C_i^{\lambda_{2i}}$ is recursively defined as
    \begin{align}
        C_n^{\lambda} &\defeq C^{\lambda},\\
        C_{i-1}^{\lambda_{2i-2}} &\defeq {1\over d_{\mcI_i}}\sum_{\substack{\lambda_{2i-1}\in \Irr{G}^{(2i-1)} \\ \lambda_{2i}\in \Irr{G}^{(2i)}}} \sum_{\substack{a\in [C_{\lambda_{2i-1}}^{\lambda_{2i}}] \\ b\in [C_{\lambda_{2i-2}}^{\lambda_{2i-1}}]}} (R_b^{\lambda_{2i-1} \to \lambda_{2i-2}} R_a^{\lambda_{2i}\to\lambda_{2i-1}}) C_i^{\lambda_{2i}} (R_b^{\lambda_{2i-1} \to \lambda_{2i-2}} R_a^{\lambda_{2i}\to\lambda_{2i-1}})^\dagger,
    \end{align}
    for $i\in [n]$.
\end{lemma}
\end{widetext}
\begin{proof}
    This lemma directly follows from Lems.~\ref{lem:orthogonal_basis_properties} and \ref{lem:tensor_and_partial_trace}, similarly to Lem.~\ref{lem:comb_condition_symmetry}.
\end{proof}

\subsection{Generalized Schur transforms and dual Schur transforms}
\label{sec:generalized_schur}
We define generalized Clebsch--Gordan (CG) transforms corresponding to a unitary representation $\rho: G\to \End(V_\rho)$ for a compact group $G$.
As shown in Eq.~\eqref{eq:decomposition_tensor_representation}, the tensor product of an irrep $\mu\in \Irr{G}$ and $\rho$ is decomposed into irreps as
\begin{align}
    V_\mu \otimes V_\rho \cong \bigoplus_{\lambda \in \Irr{G}} V_\lambda \otimes \CC^{c_{\mu\rho}^{\lambda}},
\end{align}
i.e., there exists an isomorphism  $\mathrm{CG}_{\mu, \rho}: V_\mu \otimes V_\rho \to \bigoplus_{\lambda\in \Irr{G}} V_\lambda \otimes \CC^{c_{\mu\rho}^{\lambda}}$.
The generalized CG transform is defined as a unitary operator $\widetilde{\mathrm{CG}}_{\rho}^{(n)}\defeq \bigoplus_{\mu\in \Irr{G}^{(n-1)}} \mathrm{CG}_{\mu, \rho}$, whose input and output spaces are given by
\begin{align}
    \widetilde{\mathrm{CG}}_{\rho}^{(n)}: \bigoplus_{\mu\in \Irr{G}^{(n-1)}} V_\mu \otimes V_\rho &\to \bigoplus_{\mu\in \Irr{G}^{(n-1)}} \bigoplus_{\lambda\in \Irr{G}^{(n)}} V_\lambda \otimes \CC^{c_{\mu\rho}^{\lambda}}.
\end{align}
In the quantum circuit, the input and output spaces are given by
\begin{align}
    \mathrm{CG}_{\rho}^{(n)}: &\CC^{\abs{\Irr{G}^{(n-1)}}} \otimes V_{n-1} \otimes V_\rho \to \CC^{\abs{\Irr{G}^{(n-1)}}} \otimes \CC^{C_n} \otimes \CC^{\abs{\Irr{G}^{(n)}}} \otimes V_n,
\end{align}
where $\CC^{\abs{\Irr{G}^{(n-1)}}}$ and $\CC^{\abs{\Irr{G}^{(n)}}}$ are registers storing the labels of the irreps, $V_{n-1}$ and $V_n$ are registers storing the vectors in the representation spaces, and $\CC^{C_n}$ is a register storing the multiplicity of the irreps, satisfying
\begin{align}
    \dim V_{n-1} &\geq \max_{\mu\in \Irr{G}^{(n-1)}} d_\mu,\\
    \dim V_n &\geq \max_{\lambda\in \Irr{G}^{(n)}} d_\lambda,\\
    C_n &\geq \max_{\mu\in \Irr{G}^{(n-1)}, \lambda\in \Irr{G}^{(n)}} c_{\mu\rho}^{\lambda}.
\end{align}
The superscript $n$ will be omitted when it is clear from the context.

The generalized CG transform can be represented in a quantum circuit as
\begin{align}
    \begin{split}
    \includegraphics{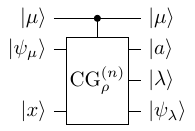}
    \end{split}.
\end{align}
By combining the generalized CG transforms, we define the \emph{generalized Schur transform} as
\begin{align}
    \begin{split}
    \includegraphics{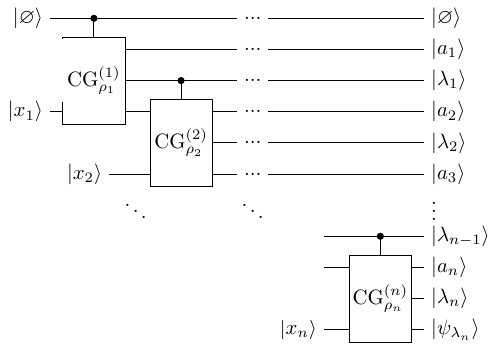}
    \end{split},
\end{align}
which implements $U_\mathrm{Sch}^{\rho_1, \ldots, \rho_n}: \bigotimes_{i=1}^{n} V_{\rho_i} \to \bigoplus_{\lambda\in \Irr{G}^{(n)}} V_\lambda \otimes \CC^{M_{\lambda}^{(n)}}$.
This is a basis change from the computational basis of $\bigotimes_{i=1}^{n} V_{\rho_i}$ to the basis of $\bigoplus_{\lambda\in \Irr{G}^{(n)}} V_\lambda \otimes \CC^{M_{\lambda}^{(n)}}$, where the basis in $\CC^{M_{\lambda}^{(n)}}$ is represented as
\begin{align}
    \ket{\varnothing}\ket{a_1}\ket{\lambda_1} \cdots \ket{a_{n}}\ket{\lambda_n},
\end{align}
which corresponds to the basis vector shown in Eq.~\eqref{eq:path_vector}.

We also introduce the CG transform corresponding to the conjugate representation $\overline{\rho}: g\in G \mapsto \overline{\rho(g)} \in \End(V_{\overline{\rho}})$, which is also called the dual CG transform.
We consider irreps $\lambda, \nu\in \Irr{G}$ and the irreducible decomposition of the tensor product of $\lambda$ and the conjugate representation $\overline{\nu}$ of $\nu$ as
\begin{align}
    V_\lambda \otimes V_{\overline{\nu}} &\cong \bigoplus_{\mu\in \Irr{G}} V_\mu \otimes \CC^{c_{\lambda\overline{\nu}}^{\mu}},
\end{align}
and consider the isomorphism $\mathrm{CG}_{\lambda, \overline{\nu}}: V_\lambda \otimes V_{\overline{\nu}} \to \bigoplus_{\mu\in \Irr{G}} V_\mu \otimes \CC^{c_{\lambda\overline{\nu}}^{\mu}}$.
By using Eq.~\eqref{eq:littlewood_richardson}, we can show that $c_{\lambda\overline{\nu}}^{\mu} = c_{\mu\nu}^{\lambda}$ holds, i.e.,
\begin{align}
    \CC^{c_{\lambda\overline{\nu}}^{\mu}} \simeq \CC^{c_{\mu\nu}^{\lambda}}.
\end{align}
We also have
\begin{align}
    V_{\nu} \simeq V_{\overline{\nu}}
\end{align}
as linear spaces\footnote{We denote $A \cong B$ to represent that $A$ and $B$ are isomorphic as representation spaces, while we denote $A\simeq B$ to represent that $A$ and $B$ are isomorphic as linear spaces.}.
We can take the orthonormal bases $\{\ket{\psi_\nu^{(i)}}\}$ of $V_{\nu} \simeq V_{\overline{\nu}}$ and $\{\ket{a}\}$ of $\CC^{c_{\lambda\overline{\nu}}^{\mu}} \simeq \CC^{c_{\mu\nu}^{\lambda}}$ such that [see Eq.~(10), p. 298 in Ref.~\cite{klimyk1995representations}]
\begin{align}
    \bra{\psi_{\mu}}\bra{a} \mathrm{CG}_{\lambda, \overline{\nu}}\ket{\psi_\lambda}\ket{\psi_{\nu}^{(i)}} = \sqrt{d_\mu \over d_\lambda}\overline{\bra{\psi_\lambda}\bra{a} \mathrm{CG}_{\mu, \nu}\ket{\psi_\mu}\ket{\psi_\nu^{(i)}}}
\end{align}
holds for all $\ket{\psi_\mu}\in V_\mu, \ket{\psi_\lambda}\in V_{\lambda}, i, a$.
In particular, we consider the trivial representation $\mu = \varnothing$ and $\nu = \lambda$.
Since $c_{\lambda \overline{\lambda}}^{\varnothing} = c_{\varnothing \overline{\lambda}}^{\lambda} = 1$, we have
\begin{align}
    \bra{\psi_{\varnothing}} \mathrm{CG}_{\lambda, \overline{\lambda}}\ket{\psi_{\lambda}^{(i)}} \ket{\psi_{\lambda}^{(j)}} &= \sqrt{1\over d_\lambda} \overline{\bra{\psi_{\lambda}^{(i)}} \mathrm{CG}_{\lambda, \overline{\lambda}}\ket{\psi_{\varnothing}} \ket{\psi_{\lambda}^{(j)}}} 
    = {\delta_{ij} \over \sqrt{d_\lambda}} \quad \forall i,j,
\end{align}
i.e.,
\begin{align}
    \mathrm{CG}_{\lambda, \overline{\lambda}}^\dagger \ket{\psi_{\varnothing}} = {1\over \sqrt{d_\lambda}} \sum_{i\in [d_\lambda]} \ket{\psi_{\lambda}^{(i)}}\otimes \ket{\psi_\lambda^{(i)}},
\end{align}
where $\ket{\psi_\varnothing}$ is a vector in $V_\varnothing\cong \CC$.
Similarly, for any (not necessarily irreducible) representation $\rho: G\to \End(V_\rho)$, we can take the orthonormal bases $\{\ket{\psi_\rho^{(i)}}\}$ of $V_{\rho} \simeq V_{\overline{\rho}}$ and $\{\ket{a}\}$ of $\CC^{c_{\lambda\overline{\rho}}^{\mu}} \simeq \CC^{c_{\mu\rho}^{\lambda}}$ such that
\begin{align}
    \bra{\psi_{\mu}}\bra{a} \mathrm{CG}_{\lambda, \overline{\rho}}\ket{\psi_\lambda}\ket{\psi_{\rho}^{(i)}} = \sqrt{d_\mu \over d_\lambda}\overline{\bra{\psi_\lambda}\bra{a} \mathrm{CG}_{\mu, \rho}\ket{\psi_\mu}\ket{\psi_\rho^{(i)}}}
\end{align}
holds for all $\ket{\psi_\mu}\in V_\mu, \ket{\psi_\lambda}\in V_{\lambda}, i, a$,
and we obtain
\begin{align}
    \mathrm{CG}_{\lambda,\overline{\rho}}^\sfT \ket{\psi_{\mu}} \ket{a}\nonumber
    &=\sum_{i\in [d_\lambda]} \sum_{j\in [d_\rho]} \ket{\psi_\lambda^{(i)}}\ket{\psi_\rho^{(j)}} \bra{\psi_\lambda^{(i)}}\bra{\psi_\rho^{(j)}} \mathrm{CG}_{\lambda,\overline{\rho}}^\sfT \ket{\psi_{\mu}} \ket{a}\\
    &= \sqrt{d_\mu \over d_\lambda} \sum_{i\in [d_\lambda]} \sum_{j\in [d_\rho]} \ket{\psi_\lambda^{(i)}}\ket{\psi_\rho^{(j)}} \overline{\bra{\psi_\lambda^{(i)}}\bra{a} \mathrm{CG}_{\mu, \rho}\ket{\psi_\mu}\ket{\psi_\rho^{(j)}}}\\
    &= \sqrt{d_\mu \over d_\lambda} \sum_{i\in [d_\lambda]} \sum_{j\in [d_\rho]} \ket{\psi_\rho^{(j)}} \otimes \overline{\bra{a} \mathrm{CG}_{\mu, \rho}\ket{\psi_\mu}\ket{\psi_\rho^{(j)}}},
    \label{eq:CG_conjugate}
\end{align}
where $\mathrm{CG}_{\lambda,\overline{\rho}}^\sfT$ is the transpose of $\mathrm{CG}_{\lambda,\overline{\rho}}^\sfT$ defined by
\begin{align}
    \bra{\psi_\lambda^{(i)}}\bra{\psi_\rho^{(j)}} \mathrm{CG}_{\lambda,\overline{\rho}}^\sfT \ket{\psi_{\mu}} \ket{a} = \bra{\psi_{\mu}} \bra{a} \mathrm{CG}_{\lambda,\overline{\rho}} \ket{\psi_\lambda^{(i)}}\ket{\psi_\rho^{(j)}},
\end{align}
and the overline represents the complex conjugate in the basis $\{\ket{\psi_\rho^{(j)}}\}_j$.
For $\mu = \varnothing$, we have
\begin{align}
    \mathrm{CG}_{\lambda,\overline{\rho}}^\sfT \ket{\psi_{\varnothing}} \ket{a}
    &= {1\over \sqrt{d_\lambda}} \sum_{j\in [d_\rho]} \ket{\psi_{\rho}^{(j)}} \otimes \overline{\bra{a}\mathrm{CG}_{\varnothing, \rho}\ket{\psi_\varnothing}\ket{\psi_\rho^{(j)}}}.
\end{align}
By composing the inverse of dual CG transforms, we define the \emph{dual Schur transform} as
\begin{align}
    \label{eq:generalied_dual_Schur}
    \begin{split}
        \includegraphics[width=0.5\linewidth]{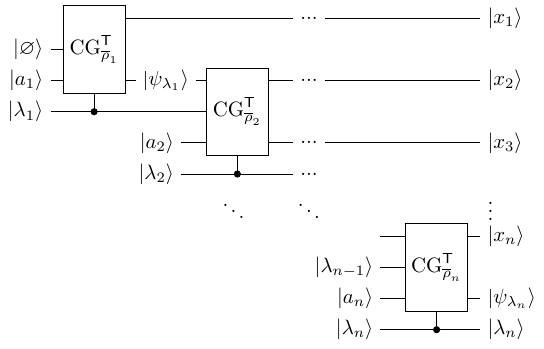}
    \end{split},
\end{align}
which implements $U_\mathrm{dSch}^{\rho_1, \ldots, \rho_n}: \bigoplus_{\lambda\in \Irr{G}^{(n)}} \CC^{M_\lambda^{(n)}} \to \bigotimes_{i=1}^{n} V_{\rho_i} \otimes \bigoplus_{\lambda\in \Irr{G}^{(n)}} V_\lambda$ such that
\begin{align}
    &U_\mathrm{dSch}^{\rho_1, \ldots, \rho_n} \ket{p_\lambda}\nonumber\\
    &= {1\over \sqrt{d_\lambda}} \sum_{\vec{x}} \ket{\vec{x}} \otimes \overline{(\bra{\lambda}\otimes \bra{p_\lambda}) U_\mathrm{Sch}^{\rho_1, \ldots, \rho_n}\ket{\vec{x}}}\\
    &= {1\over \sqrt{d_\lambda}} \sum_{i\in [d_\lambda]} \left(U_\mathrm{Sch}^{\rho_1, \ldots, \rho_n}\right)^\dagger \ket{\lambda, p_\lambda, \psi_\lambda^{(i)}} \otimes \ket{\psi_\lambda^{(i)}},
    \label{eq:dual_schur_action}
\end{align}
where $\{\ket{\vec{x}}\} = \{\ket{x_1}\otimes \cdots \otimes \ket{x_n}\}$ is the computational basis of $\bigotimes_{i=1}^{n} V_{\rho_i}$, we use the equality
\begin{align}
    \label{eq:maximally_entangled_basis_change}
    \sum_{i} \overline{\ket{\psi_i}} \otimes \ket{\psi_i} = \sum_{j} \overline{\ket{\phi_j}} \otimes \ket{\phi_j}
\end{align}
for any orthonormal bases $\{\ket{\psi_i}\}_i$ and $\{\ket{\phi_j}\}_j$ of a Hilbert space in the last line, and the overline represents the complex conjugate in the computational basis.
Note that the superscripts of $\mathrm{CG}_{\overline{\rho}_i}$ in the circuit~\eqref{eq:generalied_dual_Schur} are omitted for simplicity.

The generalized CG transforms includes CG transform ($\rho: g\in \U(d) \mapsto g \in \End(\CC^d)$) and dual CG transform ($\rho: g\in \U(d) \mapsto \overline{g} \in \End(\CC^d)$) as special cases, where $\overline{g}$ represents the complex conjugate of $g$.
The generalized Schur transform includes the quantum Schur transform and the mixed Schur transform as special cases.

\subsection{Proof of Theorem~\ref{thm:comb_implementation}}
\label{sec:proof_comb_implementation}
\begin{proof}
Using the operators $C_i^{\lambda_i\mu_i}$ defined in Lem.~\ref{lem:comb_condition_symmetry}, we define an isometry operator
\begin{align}
    V_i^{\lambda_i \mu_{i-1}}: &\bigoplus_{\lambda_{i-1}\in \Irr{G}^{(i-1)}} \Supp(C_{i-1}^{\lambda_{i-1}\mu_{i-1}}) \otimes \CC^{c_{\lambda_{i-1}\rho_i}^{\lambda_i}} \to \bigoplus_{\mu_i\in \Irr{H}^{(i)}} \Supp(C_i^{\lambda_i \mu_i}) \otimes \CC^{c_{\mu_{i-1} \sigma_i}^{\mu_i}}
\end{align}
by
\begin{multline}
    V_i^{\lambda_i \mu_{i-1}} \ket{\mathrm{mem}_{i-1}} \otimes \ket{\lambda_{i-1}} \otimes \ket{a_i}
     \defeq \sqrt{d_{\lambda_{i-1}} \over d_{\lambda_i}} \sum_{\substack{\mu_i\in \Irr{H}^{(i)}\\b_i\in [c_{\mu_{i-1}\sigma_i}^{\mu_i}]}} \ket{b_i} \otimes \ket{\mu_i} \otimes (C_i^{\lambda_i \mu_i})^{1/2} \\
    \times\left[(C_{i-1}^{\lambda_{i-1}\mu_{i-1}})^{-1/2} \ket{\mathrm{mem}_{i-1}} \otimes \ket{a_i} \otimes \ket{b_i}\right],
\end{multline}
where $\ket{\mathrm{mem}_{i-1}}\in \Supp(C_{i-1}^{\lambda_{i-1}\mu_{i-1}})$ and $a_i\in [c_{\lambda_{i-1}\rho_i}^{\lambda_i}]$.
The isometric property can be checked as follows:
\begin{widetext}
\begin{align}
    &\bra{\mathrm{mem}'_{i-1}} \otimes \bra{\lambda'_{i-1}} \otimes \bra{a'_i}V_i^{\lambda_i \mu_{i-1} \dagger}V_i^{\lambda_i \mu_{i-1}}\ket{\mathrm{mem}_{i-1}} \otimes \ket{\lambda_{i-1}} \otimes \ket{a_i}\nonumber\\
    &= \sum_{\substack{\mu_i\in \Irr{H}^{(i)}\\b_i\in [c_{\mu_{i-1}\sigma_i}^{\mu_i}]}} \left[\bra{\mathrm{mem}'_{i-1}} (C_{i-1}^{\lambda'_{i-1}\mu_{i-1}})^{-1/2} \otimes \bra{a'_i} \otimes \bra{b_i}\right] {d_{\lambda_{i-1}} \over d_{\lambda_i}} C_i^{\lambda_i \mu_i}\left[(C_{i-1}^{\lambda_{i-1}\mu_{i-1}})^{-1/2} \ket{\mathrm{mem}_{i-1}} \otimes \ket{a_i} \otimes \ket{b_i}\right]\\
    &= \sum_{\substack{\lambda''_{i-1}\in \Irr{G}^{(i-1)}\\a''_i\in [c_{\lambda''_{i-1}\rho_i}^{\lambda_i}]}} \left[\bra{\mathrm{mem}'_{i-1}} (C_{i-1}^{\lambda'_{i-1}\mu_{i-1}})^{-1/2} \otimes \bra{a'_i}\right] C_{i-1}^{\lambda''_{i-1} \mu_{i-1}}\otimes \ketbra{a''_i}{a''_i} \left[(C_{i-1}^{\lambda_{i-1}\mu_{i-1}})^{-1/2} \ket{\mathrm{mem}_{i-1}} \otimes \ket{a_i}\right]\\
    &= \braket{\mathrm{mem}'_{i-1}}{\mathrm{mem}_{i-1}} \delta_{\lambda_{i-1} \lambda'_{i-1}} \delta_{a_i a'_i},
\end{align}
\end{widetext}
where we use Lem.~\ref{lem:comb_condition_symmetry}.
We define the Hilbert space $\mcM_i$ such that $\Supp(C_i^{\lambda_i \mu_i})\subset \mcM_i$ for all $\lambda_i\in \Irr{G}^{(i)}$ and $\mu_i\in\Irr{H}^{(i)}$, which describes the register to store the state $\ket{\mathrm{mem}_i}$.
We define the Hilbert space $\mcA_i$ for $i\in \{0, \ldots, n\}$ by
\begin{align}
    \mcA_i\eqdef \mcV_{i}^{H} \otimes \mcR_{i}^{H} \otimes \mcM_i \otimes \mcR_i^G \otimes \mcV_i^G,
\end{align}
where $\mcR_i^G$ stores the label of irreps $\lambda_i\in \Irr{G}^{(i)}$, $\mcV_i^G$ stores the vectors $\ket{\psi_{\lambda_i}}\in V_{\lambda_i}$ for irreps $\lambda_i\in \Irr{G}^{(i)}$, and $\mcV_{i}^{H} $ and $\mcR_{i}^{H}$ are similarly defined for the group $H$.
We define the isometry $W_i: \mcA_{i-1}\otimes \mcI_i \to \mcA_i \otimes \mcO_i$ by the quantum circuit shown in Fig.~\ref{fig:comb_implementation}~(b).
Using the isometry $W_i$, the quantum comb is implemented as shown in Fig.~\ref{fig:comb_implementation}~(a).
We can check that the circuit shown in Fig.~\ref{fig:comb_implementation}~(a) and (b) implements the quantum comb~\eqref{eq:decomposition_C} as follows.
The Choi matrix $C'$ of the quantum comb given by the circuit shown in Fig.~\ref{fig:comb_implementation}~(a) and (b) is given by
\begin{align}
    C' = \Tr_{\mcA_{n}} [\dketbra{W^{(n)}}_{\mcI^n \mcO^n \mcA_n}],
\end{align}
where $\dket{W^{(n)}}$ is defined by
\begin{align}
    \dket{W^{(n)}}\defeq \sum_{\vec{x}} \ket{\vec{x}}_{\mcI^n} \otimes (W^{(n)}\ket{\vec{x}})_{\mcO^n \mcA_n}
\end{align}
using the computational basis $\{\ket{\vec{x}}\}$ of $\mcI^n$, and $W^{(n)}$ is given by
\begin{align}
\begin{split}
    \includegraphics{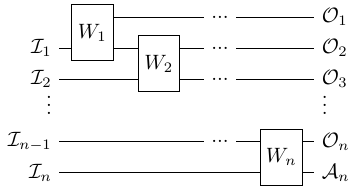}
\end{split}.
\end{align}
Since $V_i^{\lambda_i \mu_{i-1}}$ and $\mathrm{CG}_{\rho_j}$ for $j>i$ and $V_i^{\lambda_i \mu_{i-1}}$ and $\mathrm{CG}_{\overline{\sigma}_j}^\dagger$ for $j<i$ commute, $W$ is given by
\begin{widetext}
\begin{align}
    \begin{split}
        \includegraphics[width=\linewidth]{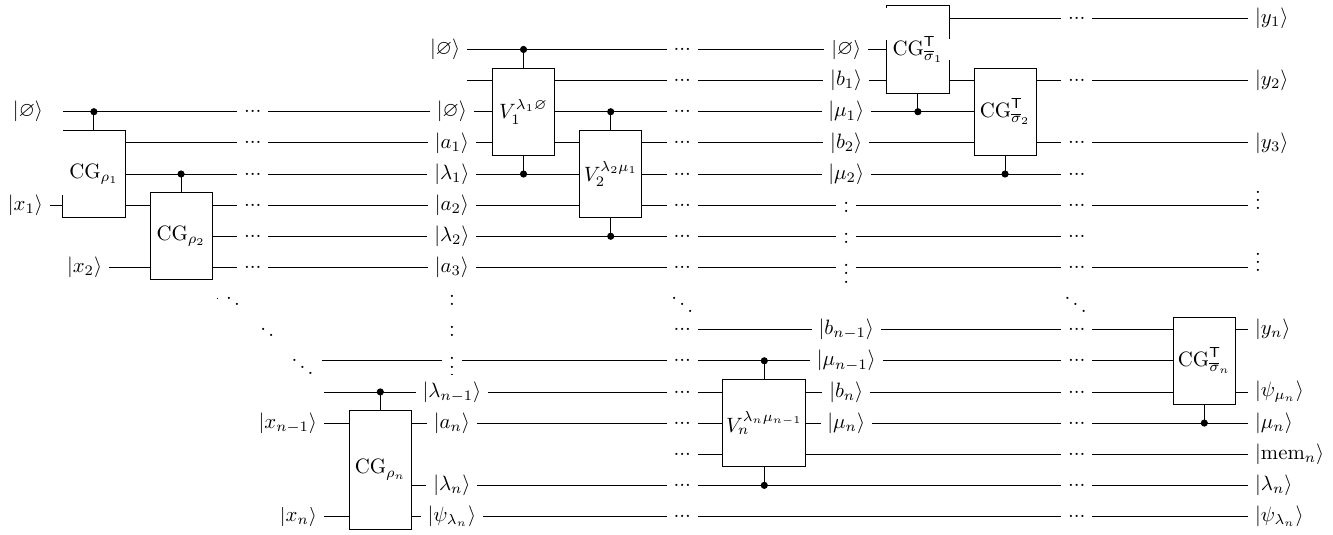}
    \end{split},
\end{align}
i.e.,
\begin{align}
    &W^{(n)} \left(U_\mathrm{Sch}^{\rho_1, \ldots, \rho_n}\right)^\dagger(\ket{\lambda, p_{\lambda}, \psi_{\lambda}^{(i)}})\nonumber\\
    &= \sum_{\mu\in \Irr{H}^{(n)}} \sum_{q_{\mu}\in \Paths(\mu, \scB_R)} U_\mathrm{dSch}^{\sigma_1, \ldots, \sigma_n} \ket{q_\mu} \otimes {1\over \sqrt{d_\lambda}} \left(C^{\lambda \mu}\right)^{1/2} (\ket{p_{\lambda}} \otimes \ket{q_{\mu}}) \otimes \ket{\psi_\lambda^{(i)}}\\
    &= \sum_{\mu\in \Irr{H}^{(n)}} \sum_{q_{\mu}\in \Paths(\mu, \scB_R)} {1\over \sqrt{d_\lambda d_\mu}} \sum_{j\in [d_\mu]} \left(U_\mathrm{Sch}^{\sigma_1, \ldots, \sigma_n}\right)^\dagger \ket{\mu, q_\mu, \psi_\mu^{(j)}} \otimes \ket{\psi_\mu^{(j)}} \otimes \left(C^{\lambda \mu}\right)^{1/2} (\ket{p_{\lambda}} \otimes \ket{q_{\mu}}) \otimes \ket{\psi_\lambda^{(i)}},
\end{align}
where we use Eq.~\eqref{eq:dual_schur_action} in the last line.
Thus, the vector $\dket{W^{(n)}}$ is given by
\begin{align}
    \dket{W^{(n)}}
    &= \sum_{\vec{x}} \ket{\vec{x}} \otimes W^{(n)}\ket{\vec{x}}\\
    &= \sum_{\lambda \in \Irr{G}^{(n)}} \sum_{p_{\lambda} \in \Paths(\lambda, \scB_L)}\sum_{i\in [d_\lambda]} \overline{\left(U_\mathrm{Sch}^{\rho_1, \ldots, \rho_n}\right)^\dagger \ket{\lambda, p_{\lambda}, \psi_{\lambda}^{(i)}}} \otimes W^{(n)} \left(U_\mathrm{Sch}^{\rho_1, \ldots, \rho_n}\right)^\dagger \ket{\lambda, p_{\lambda}, \psi_{\lambda}^{(i)}},
\end{align}
where we use Eq.~\eqref{eq:maximally_entangled_basis_change} in the last line.
Therefore, the Choi matrix is given by
\begin{align}
    C'
    &= \sum_{\substack{\lambda \in \Irr{G}^{(n)}\\\mu\in \Irr{H}^{(n)}}} \sum_{\substack{p_{\lambda}, p'_\lambda \in \Paths(\lambda, \scB_L)\\q_{\mu}. q'_\mu\in \Paths(\mu, \scB_R)}}\sum_{\substack{i, i'\in [d_\lambda]\\j, j'\in [d_\mu]}} {\Tr\left[\ketbra{\psi_\lambda^{(i)}}{\psi_\lambda^{(i')}} \otimes \ketbra{\psi_\mu^{(j)}}{\psi_\mu^{(j')}} \otimes \left(C^{\lambda \mu}\right)^{1/2} (\ketbra{p_{\lambda}}{p'_\lambda} \otimes \ketbra{q_{\mu}}{q'_\mu}) \left(C^{\lambda \mu}\right)^{1/2}\right] \over d_\lambda d_\mu}\nonumber\\
    &\hspace{30pt} \times \overline{\left(U_\mathrm{Sch}^{\rho_1, \ldots, \rho_n}\right)^\dagger \ketbra{\lambda, p_\lambda, \psi_\lambda^{(i)}}{\lambda, p'_\lambda, \psi_\lambda^{(i')}} \left(U_\mathrm{Sch}^{\rho_1, \ldots, \rho_n}\right)} \otimes \left(U_\mathrm{Sch}^{\sigma_1, \ldots, \sigma_n}\right)^\dagger\ketbra{\mu, q_\mu, \psi_\mu^{(j)}}{\mu, q'_\mu, \psi_\mu^{(j)}} \left(U_\mathrm{Sch}^{\sigma_1, \ldots, \sigma_n}\right)\\
    &= \sum_{\substack{\lambda \in \Irr{G}^{(n)}\\\mu\in \Irr{H}^{(n)}}} \sum_{\substack{p_{\lambda}, p'_\lambda \in \Paths(\lambda, \scB_L)\\q_{\mu}. q'_\mu\in \Paths(\mu, \scB_R)}} C^{\lambda\mu}_{p_\lambda q_\mu, p'_\lambda q'_\mu} {\overline{E}^{\lambda}_{p_\lambda p'_\lambda} \over d_\lambda} \otimes {\widetilde{E}^{\mu}_{q_\mu q'_\mu}\over d_\mu}\\
    &= C.
\end{align}
\end{widetext}
\end{proof}

\subsection{Proof of Theorem~\ref{thm:comb_implementation_covariant}}
\label{sec:proof_comb_implementation_covariant}
\begin{proof}
Using the operators $C_i^{\lambda_{2i}}$ defined in Lem.~\ref{lem:comb_condition_covariance}, we define an isometry operator
\begin{align}
    V_i^{\lambda_{2i-1}}: &\bigoplus_{\lambda_{2i-2}\in \Irr{G}^{(2i-2)}} \Supp(C_{i-1}^{\lambda_{2i-2}}) \otimes \CC^{c_{\lambda_{2i-2}\rho_{2i-1}}^{\lambda_{2i-1}}} \to \bigoplus_{\lambda_{2i}\in \Irr{G}^{(2i)}} \Supp(C_i^{\lambda_{2i}}) \otimes \CC^{c_{\lambda_{2i-1} \rho_{2i}}^{\lambda_{2i}}}
\end{align}
by
\begin{multline}
    V_i^{\lambda_{2i-1}}\ket{\mathrm{mem}_{i-1}} \otimes\ket{\lambda_{2i-2}} \otimes \ket{a_{2i-1}} \defeq \\ \sqrt{d_{\lambda_{2i-2}} \over d_{\lambda_{2i-1}}} \sum_{\substack{\lambda_{2i}\in \Irr{G}^{(2i)}\\a_{2i}\in [c_{\lambda_{2i-1}\rho_{2i}}^{\lambda_{2i}}]}} (C_i^{\lambda_{2i}})^{1/2} \Big[(C_{i-1}^{\lambda_{2i-2}})^{-1/2}
    \times \ket{\mathrm{mem}_{i-1}}\otimes \ket{a_{2i-1}} \otimes \ket{a_{2i}}\Big]
    \otimes \ket{\lambda_{2i}} \otimes \ket{a_{2i}},
\end{multline}
where $\ket{\mathrm{mem}_{i-1}}\in \Supp(C_{i-1}^{\lambda_{2i-2}})$ and $a_{2i-1}\in [c_{\lambda_{2i-2}\rho_{2i-1}}^{\lambda_{2i-1}}]$.
The isometric property can be checked similarly to the proof of Thm.~\ref{thm:comb_implementation}.
We define the Hilbert space $\mcM_i$ such that $\Supp(C_i^{\lambda_{2i}})\subset \mcM_i$ for all $\lambda_{2i}\in \Irr{G}^{(2i)}$, which describes the register to store the state $\ket{\mathrm{mem}_i}$.
We define the Hilbert space $\mcA_i$ for $i\in \{0, \ldots, n\}$ by
\begin{align}
    \mcA_i\defeq \mcM_i \otimes \mcR_{2i}^{G} \otimes \mcV_{2i}^{G},
\end{align}
where $\mcR_{2i}^G$ stores the label of irreps $\lambda_{2i}\in \Irr{G}^{(2i)}$ and $\mcV_{2i}^G$ stores the vectors $\ket{\psi_{\lambda_{2i}}}\in V_{\lambda_{2i}}$ for irreps $\lambda_{2i}\in \Irr{G}^{(2i)}$.
We define the isometry $W_i: \mcA_{i-1}\otimes \mcI_{i}  \to \mcA_i \otimes \mcO_{i}$ by the quantum circuit shown in Fig.~\ref{fig:comb_implementation}~(c).

We can check that the circuit shown in Fig.~\ref{fig:comb_implementation}~(a) and (c) implements the quantum comb~\eqref{eq:decomposition_C_G} by showing the following lemma.
\begin{lemma}
    \label{lem:dket_W}
    Defining $W^{(n)}$ and $\dket{W^{(n)}}$ as in the proof of Thm.~\ref{thm:comb_implementation}, the vector $\dket{W^{(n)}}$ is given by
    \begin{align}
        \dket{W^{(n)}} = \sum_{\substack{\lambda\in \Irr{G}^{(2n)}\\p_{\lambda} \in \Paths(\lambda, \scB)\\i\in [d_\lambda]}} & \left[\overline{\left(U_\mathrm{Sch}^{\rho_1, \ldots, \rho_{2n}}\right)^\dagger(\ket{\lambda, p_{\lambda}, \psi_{\lambda}^{(i)}})}\right]_{\mcI^n\mcO^n}
        \otimes {1\over \sqrt{d_\lambda}}\left[(C^{\lambda})^{1/2} \ket{p_{\lambda}} \otimes \ket{\psi_\lambda^{(i)}}\right]_{\mcA_n}.
    \end{align}
\end{lemma}
\begin{proof}[Proof of Lem.~\ref{lem:dket_W}]
    We can show this lemma by induction.
    For $n=0$, we have $\dket{W^{(0)}} = 1$ and the lemma holds.
    Suppose the lemma holds for $n-1$.
    Then, we have
    \begin{widetext}
    \begin{align}
        \dket{W^{(n)}} &= \sum_{x_n} \ket{x_n}_{\mcI_n} \otimes (\1_{\mcI^{n-1}\mcO^{n-1}} \otimes W_n) (\dket{W^{(n-1)}}_{\mcI^{n-1}\mcO^{n-1}\mcA_{n-1}} \otimes \ket{x_n})\\
        &= \sum_{\substack{\lambda_{2n-2}\in \Irr{G}^{(2n-2)}\\p_{\lambda_{2n-2}}\in \Paths(\lambda_{2n-2}, \scB)}} \sum_{i\in [d_{\lambda_{2n-2}}]} \sum_{x_n} {1\over \sqrt{d_{\lambda_{2n-2}}}} \ket{x_n}_{\mcI_n} \otimes (\1_{\mcI^{n-1}\mcO^{n-1}} \otimes W_n) \nonumber\\
        &\hspace{15pt} \times \left[\overline{\left(U_\mathrm{Sch}^{\rho_1, \ldots, \rho_{2n-2}}\right)^\dagger \ket{\lambda_{2n-2}, p_{\lambda_{2n-2}}, \psi_{\lambda_{2n-2}}^{(i)}}}\right]_{\mcI^{n-1}\mcO^{n-1}}\nonumber\\ 
        &\hspace{15pt} \otimes \left[(C_{n-1}^{\lambda_{2n-2}})^{1/2} \ket{p_{\lambda_{2n-2}}} \otimes \ket{\lambda_{2n-2}} \otimes \ket{\psi_{\lambda_{2n-2}}^{(i)}}\right]_{\mcA_{n-1}} \otimes \ket{x_n}\\
        &= \sum_{\substack{\lambda_{2n-2}\in \Irr{G}^{(2n-2)}\\p_{\lambda_{2n-2}}\in \Paths(\lambda_{2n-2}, \scB)}} \sum_{i\in [d_{\lambda_{2n-1}}]} \sum_{\lambda_{2n-1} \in \Irr{G}^{(2n-1)}} \sum_{a_{2n-1}\in [c_{\lambda_{2n-2}\rho_{2n-1}}^{\lambda_{2n-1}}]} \nonumber\\
        &\hspace{30pt} \times\left[\overline{\left(U_\mathrm{Sch}^{\rho_1, \ldots, \rho_{2n-1}}\right)^\dagger\ket{\lambda_{2n-1}, p_{\lambda_{2n-2}} \xrightarrow{a_{2n-1}} \lambda_{2n-1}, \psi_{\lambda_{2n-1}}^{(i)}}}\right]_{\mcI^{n}\mcO^{n-1}}\nonumber\\
        &\hspace{30pt} \otimes \left[\1_{\mcM_n} \otimes {\mathrm{CG}}_{\lambda_{2n-1}, \overline{\rho}_{2n}}^\dagger\right]V_i^{\lambda_{2n-1}} {1\over \sqrt{d_{\lambda_{2n-2}}}} \left[(C_{n-1}^{\lambda_{2n-2}})^{1/2} \ket{p_{\lambda_{2n-2}}} \otimes \ket{a_{2n-1}}\right] \otimes \ket{\psi_{\lambda_{2n-1}}^{(i)}}\\
        &= \sum_{\substack{\lambda_{2n-2}\in \Irr{G}^{(2n-2)}\\p_{\lambda_{2n-2}}\in \Paths(\lambda_{2n-2}, \scB)}} \sum_{i\in [d_{\lambda_{2n-1}}]} \sum_{\lambda_{2n-1} \in \Irr{G}^{(2n-1)}} \sum_{a_{2n-1}\in [c_{\lambda_{2n-2}\rho_{2n-1}}^{\lambda_{2n-1}}]} \sum_{\lambda_{2n}\in \Irr{G}^{(2n)}} \sum_{a_{2n}\in [c_{\lambda_{2n-1}\rho_{2n}}^{\lambda_{2n}}]} \nonumber\\
        &\hspace{30pt} \times\left[\overline{\left(U_\mathrm{Sch}^{\rho_1, \ldots, \rho_{2n-1}}\right)^\dagger\ket{\lambda_{2n-1}, p_{\lambda_{2n-2}} \xrightarrow{a_{2n-1}} \lambda_{2n-1}, \psi_{\lambda_{2n-1}}^{(i)}}}\right]_{\mcI^{n}\mcO^{n-1}}\nonumber\\
        &\hspace{30pt} \otimes \left[(C^{\lambda_{2n}})^{1/2} \ket{p_{\lambda_{2n-2}} \xrightarrow{a_{2n-1}} \lambda_{2n-1} \xrightarrow{a_{2n}} \lambda_{2n}}\right]_{\mcM_n} \otimes {1 \over \sqrt{d_{\lambda_{2n-1}}}} \left[\mathrm{CG}_{\lambda_{2n}, \overline{\rho}_{2n}}^\dagger (\ket{\psi_{\lambda_{2n-1}}^{(i)}} \otimes \ket{a_{2n}})\right]_{\mcV_{2n}^{G} \mcO_n}\\
        &= \sum_{\substack{\lambda_{2n-1}\in \Irr{G}^{(2n-1)}\\p_{\lambda_{2n-1}}\in \Paths(\lambda_{2n-1}, \scB)}} \sum_{i\in [d_{\lambda_{2n-1}}]}  \sum_{\lambda_{2n}\in \Irr{G}^{(2n)}} \sum_{a_{2n}\in [c_{\lambda_{2n-1}\rho_{2n}}^{\lambda_{2n}}]} \sum_{j\in [d_{\lambda_{2n}}]} \sum_{y_n}\nonumber\\
        &\hspace{30pt} \times \left[\overline{\left(U_\mathrm{Sch}^{\rho_1, \ldots, \rho_{2n-1}}\right)^\dagger\ket{\lambda_{2n-1}, p_{\lambda_{2n-1}}, \psi_{\lambda_{2n-1}}^{(i)}}}\right]_{\mcI^n\mcO^{n-1}} \otimes \left[(C^{\lambda_{2n}})^{1/2} \ket{p_{\lambda_{2n-1}} \xrightarrow{a_{2n}} \lambda_{2n}}\right]_{\mcM_n}\nonumber\\
        &\hspace{30pt} \otimes {1\over \sqrt{d_{\lambda_{2n}}}} \ket{y_n}_{\mcO_n} \otimes \left[\bra{a_{2n}} \mathrm{CG}_{\lambda_{2n-1}, \rho_{2n}} \ket{\psi_{\lambda_{2n-1}}^{(i)}} \ket{y_n}\right]_{\mcV_{2n}^G}\\
        &= \sum_{\substack{\lambda_{2n}\in \Irr{G}^{(2n)}\\p_{\lambda_{2n}}\in \Paths(\lambda_{2n}, \scB)}} \sum_{j\in [d_{\lambda_{2n}}]} \left[\overline{\left(U_\mathrm{Sch}^{\rho_1, \ldots, \rho_{2n}}\right)^\dagger\ket{\lambda_{2n}, p_{\lambda_{2n}}, \psi_{\lambda_{2n}}^{(j)}}}\right]_{\mcI^n\mcO^n} \otimes {1\over \sqrt{d_{\lambda_{2n}}}}\left[(C^{\lambda_{2n}})^{1/2} \ket{p_{\lambda_{2n}}} \otimes \ket{\psi_{\lambda_{2n}}^{(j)}}_{\mcV_{2n}^G}\right]_{\mcA_n},
    \end{align}
    \end{widetext}
    where we use the equality~\eqref{eq:maximally_entangled_basis_change} in the third and last equalities, and Eq.~\eqref{eq:CG_conjugate} in the fifth equality.
    Thus, the lemma holds for $n$, which concludes the proof.
\end{proof}
From Lem.~\ref{lem:dket_W}, the Choi matrix $C'$ of the quantum comb given by the circuit shown in Fig.~\ref{fig:comb_implementation}~(a) and (c) is given by
\begin{align}
    C' &= \Tr_{\mcA_{n}} [\dketbra{W^{(n)}}_{\mcI^n \mcO^n \mcA_n}]\\
    &= \sum_{\lambda\in \Irr{G}^{(2n)}} \sum_{p_{\lambda}, p'_\lambda \in \Paths(\lambda, \scB)} C^{\lambda}_{p_\lambda, p'_\lambda} {\overline{E}^{\lambda}_{p_\lambda p'_\lambda} \over d_\lambda} \\
    &= C,
\end{align}
which concludes the proof.
\end{proof}

\section{SDP of deterministic quantum combs for unitary transposition and inversion}

Following the work of \cite{quintino2022deterministic} and others \cite{quintino2019reversing,quintino2019probabilistic,yoshida2021universal,yoshida2023reversing,ComplexConjugation}, we consider a task of universal transformation of a black-box unitary operation. 
Consider the following general problem: given $n$ copies of an unknown $d$-dimensional unitary $U$, the task is to find a universal protocol that implements $f(U)$, where $f$ is some function of $U$. 
This protocol can be either \emph{deterministic} or \emph{probabilistic}, depending on whether it always succeeds or not, and either \emph{exact} or \emph{non-exact}, depending on the channel fidelity between the ideal channel and the one implemented by the protocol. 
More specifically, we consider the problem of finding optimal deterministic protocols for $f(U)=U\tp$ and $f(U) = U^{-1}$.

Following previous work, we use the formalism of quantum combs or sequential superchannels \cite{quintino2022deterministic}. 
A \emph{quantum superchannel} is a linear map $\mathcal{C}\colon \bigotimes_{i=1}^{n} \of[\big]{\End(\I_i) \rightarrow \End(\O_i)} \rightarrow \of[\big]{\End(\P) \rightarrow \End(\F)}$\footnote{Note that the notation for $\I$ and $\O$ used in the main text and Appendix~\ref{sec:comb_tensor_representation} of this manuscript is flipped, and in this specific section we use more standard notation from \cite{quintino2022deterministic} for past, future, input and output of a quantum comb.
The notation here becomes consistent with the main text and Appendix~\ref{sec:comb_tensor_representation} by transforming $n+1\to n$, $\mcI_i \to \mcO_{i}$, $\mcO_i \to \mcI_{i+1}$, $\mcP\to \mcI_1$, and $\mcF\to \mcO_n$.}
that transforms $n$ quantum channels into a new quantum channel.
Here the spaces $\I_i = \O_i = \P = \F = \C^d$ correspond to the inputs $\I_i$ and outputs $\O_i$ of the $i$-th copy of the channel $\mathcal{U}(\rho) \defeq U \rho U\ct$ associated with the unknown input unitary $U$, and $\P$ and $\F$ are the input and output spaces of the desired output channel $\mathcal{U}_f(\rho) := f(U) \rho f(U)\ct$ that represents the target unitary $f(U)$. 
Let $\I^n \defeq \bigotimes_{i=1}^n \I_i$ and $\O^n \defeq \bigotimes_{i=1}^n \O_i$.
A \emph{quantum sequential superchannel} $\mathcal{C}$ (also known as a \emph{quantum comb}) is a quantum superchannel with certain additional causal constraints on its Choi matrix $C \in \End(\P \otimes \I^n \otimes \O^n \otimes \F)$ given by Eq.~\eqref{eq:comb_condition}.

In the following, it is important to understand representation theory of partially transposed permutation matrix algebras $\A^d_{n,m}$ for $m \in \set{0,1}$, since symmetry reduction of our SDPs will be done using irreps of this algebra.
This algebra appears in the mixed Schur--Weyl duality and has been extensively studied in Refs.~\cite{studzinski2013commutant,mozrzymas2014structure,mozrzymas2018simplified,grinko2024linear,grinko2023gelfand,grinko2025mixed,studzinski2025group,horodecki2026iterative}.
A special case $\A^d_{n,0}$ of this algebra corresponds to the matrix symmetric group algebra.
The representation theory of $\A^d_{n,1}$ is governed to a large extent by a Bratteli diagram, which encodes specific choices of bases of irreducible representations.
Different levels of the Bratteli diagram correspond to irreducible representations of the corresponding subalgebra, while arrows indicate how representations restrict.
In the following, we consider two different bases: left (L) and right (R). 
Their Bratteli diagrams $\scB_L$ and $\scB_R$ are shown in \cref{fig:brat_r,fig:brat_l}. In both tasks, we successfully reproduce the known results from Refs.~\cite {quintino2022deterministic,yoshida2023reversing,grinko2024linear}, while obtaining a range of new values.

\begin{figure}[!ht]
    \centering
    \begin{minipage}[t]{0.475\textwidth}
        \centering
        \includegraphics[width=\textwidth, page=2]{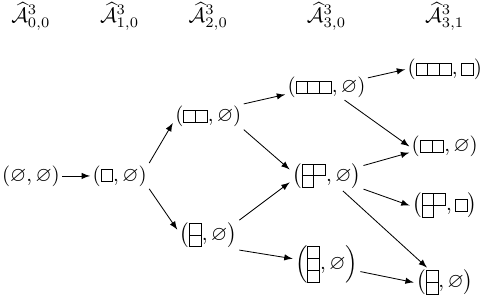}
        \caption{Example of a Bratteli diagram $\scB_L$ for $\A^3_{n,1}$ adapted to $\A^d_{0,0} \hookrightarrow \A^d_{1,0} \hookrightarrow \A^d_{2,0} \hookrightarrow \A^d_{3,0} \hookrightarrow \A^d_{3,1}$.
        Matrix units $E^{\lambda}_{T,T'}$ from \cref{6:def:Omega,eq:app4_ansatz_C} are adapted to this chain.}
        \label{fig:brat_l}
    \end{minipage}
    \hfill
    \begin{minipage}[t]{0.475\textwidth}
        \centering
        \includegraphics[width=\textwidth, page=3]{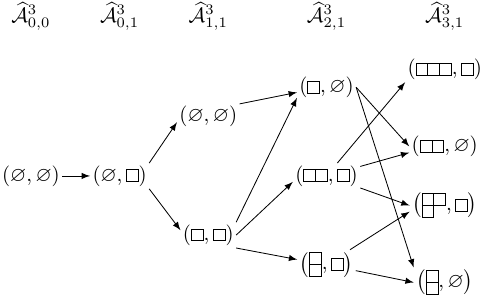}
        \caption{Example of a Bratteli diagram $\scB_R$ for $\A^3_{n,1}$ adapted to $\A^d_{0,0} \hookrightarrow \A^d_{0,1} \hookrightarrow \A^d_{1,1} \hookrightarrow \A^d_{2,1} \hookrightarrow \A^d_{3,1}$.
        Matrix units $\widetilde{E}^{\mu}_{Q,Q'}$ from \cref{eq:app4_ansatz_C} are adapted to this chain.}
        \label{fig:brat_r}
    \end{minipage}
\end{figure}

\subsection{Unitary transposition}
Finding a deterministic sequential superchannel $\mathcal{C}$ which implements the operation $\mathcal{C} ( \mathcal{U}\xp{n}) = \mathcal{U}\tp$ with highest possible average channel fidelity is equivalent to solving the following SDP for the Choi matrix $C$ of $\mathcal{C}$ \cite{quintino2022deterministic}:
\begin{equation}
  \begin{aligned}
    \max_{C} \quad & \Tr \of{C \Omega_{n,d}} \\
    \textrm{s.t.} \quad & C \text{ satisfies Eq.~\eqref{eq:comb_condition}},
  \end{aligned}
  \label{eq:black-box unitary SDP}
\end{equation}
where $\Omega_{n,d}$ is given by
\begin{equation}
    \label{6:def:Omega}
    \Omega_{n,d} \defeq \frac{1}{d^2} \sum_{\lambda \in \Irr{\A^d_{n,1}}} \sum_{ T,T' \in \Paths(\lambda,\scB_L)} \frac{ \of{E^{\lambda}_{T,T'}}_{\I^n \F} \otimes \of{E^{\lambda}_{T,T'}}_{\O^n \P}}{d_\lambda},
\end{equation}
where $E^{\lambda}_{T,T'}$ are matrix units for the Gelfand--Tsetlin basis of $\A^d_{n,1}$, adapted to the sequence $\A^d_{0,0} \hookrightarrow \A^d_{1,0} \hookrightarrow \dotsc \hookrightarrow \A^d_{n,0} \hookrightarrow \A^d_{n,1}$, and $\scB_L$ is the Bratteli diagram corresponding to this sequence of algebras \cite{grinko2025quantum}, see \cref{fig:brat_l}.
Notice that $\Omega_{n,d}$ has the mixed unitary symmetry:
\begin{equation}
    \sof[\big]{\Omega_{n,d}, V\xp{n}_{\I^n} \x \bar{V}_{\F} \x U\xp{n}_{\O^n} \x \bar{U}_{\P} } = 0, \qquad \forall U, V \in \U_d.
    \label{eq:app4_symmetry_Omega}
\end{equation}
Therefore without loss of generality the optimal solution of the SDP (\ref{eq:black-box unitary SDP}) also has the same symmetry:
\begin{equation}
    \sof[\big]{C,  V\xp{n}_{\I^n} \x \bar{V}_{\F} \x U\xp{n}_{\O^n} \x \bar{U}_{\P}} = 0, \qquad \forall U, V \in \U_d,
    \label{eq:app4_symmetry_C}
\end{equation}
which allows us to use the following ansatz for $C$:
\begin{equation}
    C = \sum_{\lambda,\mu \in \Irr{\A^d_{n,1}}} \sum_{S,S' \in \Paths(\lambda,\scB_L)} \sum_{Q,Q' \in \Paths(\mu,\scB_R)} c^{\lambda \mu}_{SS'QQ'} \frac{\of{E^{\lambda}_{S,S'}}_{\I^n \F}}{d_\lambda} \otimes \frac{\of{\widetilde{E}^{\mu}_{Q,Q'}}_{ \P \O^n}}{d_\mu},
    \label{eq:app4_ansatz_C}
\end{equation}
where $\widetilde{E}^{\lambda}_{Q,Q'}$ are matrix units for the Gelfand--Tsetlin basis of $\A^d_{n,1}$, adapted to a different sequence $\A^d_{0,0} \hookrightarrow \A^d_{0,1} \hookrightarrow \A^d_{1,1} \hookrightarrow \dotsc \hookrightarrow \A^d_{n,1}$, and $\scB_R$ is the Bratteli diagram corresponding to it, , see \cref{fig:brat_r}.
The reason we choose a different Gelfand--Tsetlin basis on the systems $\P \O^n$ is that this choice is more suitable for simplification of the partial trace constraints in \cref{eq:comb_condition}.

Note that the semidefinite constraint in \eqref{eq:comb_condition} becomes
\begin{equation}
    C \succeq  0 \qquad \Leftrightarrow \qquad \sof[\big]{c^{\lambda \mu}_{SS'QQ'}}_{(SQ),(S'Q')} \succeq  0, \quad \forall \lambda, \mu \in \Irr{\A^d_{n,1}},
\end{equation}
where we think of $\sof[\big]{c^{\lambda \mu}_{SS'QQ'}}_{(SQ),(S'Q')} \in \End(\C^{d_\lambda} \otimes \C^{d_\mu})$ as matrices. 

Using \cref{6:def:Omega,eq:app4_ansatz_C} we can rewrite the objective function as
\begin{equation}
    \Tr \of{C \Omega_{n,d}} = \frac{1}{d^2} \sum_{\lambda,\mu \in \Irr{\A^d_{n,1}}} \sum_{\substack{T,T' \in \Paths(\lambda,\scB_L) \\ Q,Q' \in \Paths(\mu,\scB_R)}} \frac{c^{\lambda,\mu}_{T'TQQ'}}{d_\lambda d_\mu} \Tr \of{ \psi^d_{n+1}(\pi) E^{\lambda}_{TT'} \psi^d_{n+1}(\pi^{-1}) \widetilde{E}^{\mu}_{TT'}},
\end{equation}
where $\psi^d_{n+1}(\pi)$ is the tensor representation of the full cyclic permutation on $(\C^d)^{n+1}$. 
The coefficients $\Tr \of{ \psi^d_{n+1}(\pi) E^{\lambda}_{TT'} \psi^d_{n+1}(\pi^{-1}) \widetilde{E}^{\mu}_{TT'}}$ could be computed numerically using Clebsch--Gordan tensor network representation of matrix units \cite{grinko2023gelfand}.

Finally, we can rewrite non-signaling constraints of the SDP (\ref{eq:black-box unitary SDP}) in the Gelfand--Tsetlin basis by using \cref{lem:tensor_and_partial_trace}. 
First, we define the following parametrisation of $C_i$ operators:
\begin{align}
    C_i = \sum_{\substack{\lambda \in \scB_L^{(i)} \\ \mu \in \scB_R^{(i)}}} \sum_{\substack{S,S' \in \Paths_i(\lambda,\scB_L) \\ Q,Q' \in \Paths_i(\mu,\scB_R)}} c^{i, \lambda \mu}_{SS'QQ'} \frac{\of{E^{\lambda}_{S,S'}}_{\I^i}}{d_\lambda} \otimes \frac{\of{\widetilde{E}^{\mu}_{Q,Q'}}_{\P \O^{i-1}}}{d_\mu},
\end{align}
where $\scB^{(i)}$ denotes the set of irreps at the level $i$ of a Bratteli diagram $\scB$.
Then we get the following rewriting of the constraints (\ref{eq:black-box unitary SDP}):
\begin{align}
    c^{i-1, \lambda \rho}_{TT'RR'} &= \sum_{\nu : \lambda \to \nu} \sum_{\theta: \rho \to \theta} c^{i, \nu \theta}_{T\to\nu, T'\to\nu, R\to\theta, R'\to\theta} \\
    \sum_{\nu : \lambda \to \nu} c^{i, \nu \mu}_{T\to\nu, T'\to\nu, Q, Q'} &= \sum_{\omega : \omega \to \mu} \delta_{Q_{i-1},\omega} \delta_{Q'_{i-1},\omega} c^{i-1, \lambda \omega}_{TT'\bar{Q}\bar{Q'}} \frac{d_\mu}{d \cdot d_\omega} 
\end{align}
for every $i \in [n+1]$, $\lambda \in \scB_L^{(i-1)}$, $\mu \in \scB_R^{(i)}$, $\rho \in  \scB_R^{(i-1)}$, $T,T' \in \Paths_{i-1}(\lambda, \scB_L)$, $R,R' \in \Paths_{i-1}(\rho,\scB_R)$ and $Q,Q' \in \Paths_{i}(\mu,\scB_R)$.
Note that this set of equations is a special case of more general \cref{lem:comb_condition_symmetry}.

With this simplification of constraints we can solve SDP numerically in Julia. 
Our numerical results are summarised in \cref{table:sdp}.

\subsection{Unitary inversion}
For completeness, we comment on the SDP for inversion case, which was first derived in \cite{yoshida2023reversing}.
In that case, the SDP \eqref{eq:black-box unitary SDP} has the same form, except that the matrices $C$ and $\Omega_{n,d}$ possess a different symmetry: they commute with $U\xp{n+1}_{\O^n \P} \x V\xp{n+1}_{\I^n\F}$ for every $U, V \in \U_d$. 
This corresponds to the case $m=0$ in the formalism of mixed Schur--Weyl duality. 
The only difference with the transposition case is that now both sets of matrix units $\of{E^{\lambda}_{S,S'}}_{\I^n \F}$ and $\of{\widetilde{E}^{\mu}_{Q,Q'}}_{ \P \O^n}$ are coming from the matrix permutation algebra $\A^d_{n+1,0}$.
The Bratteli diagrams $\scB_L$ and $\scB_R$ for these matrix units are the same, and it is simply a Young lattice with paths corresponding to standard Young tableaux, \cref{fig:young_lattice}.
Finally, our numerical results for inversion case are also summarised in \cref{table:sdp}.

\begin{figure}[!ht]
    \centering
    \includegraphics[width=0.5\textwidth, page=1]{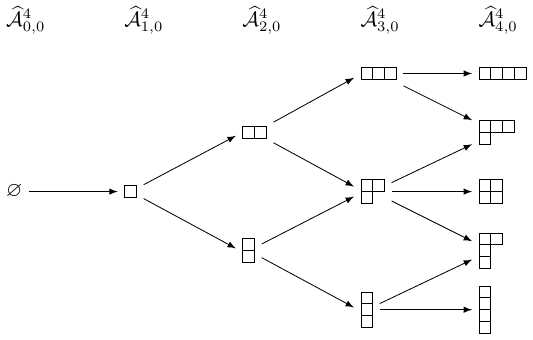}
    \caption{Bratteli diagram for symmetric group (Young lattice). Example for $\CS_4 \cong \A^4_{4,0}$.}
    \label{fig:young_lattice}
\end{figure}

%%%%%%%%%%%%%%%%%%%%%%%%%%%%%%%%%%%%%%%%%%%%%%%%%%%%%%%%%%%%%%%%%%%%%%%%%%%%%%%%%%%%%%%%%%%%%%%%%%%%%%%%%%%%%%%%%%%%%%%%%%%%%%%%%%%%%%%%%%%%%%%%%%%%%%%%%%%%%%%%%%%%%%%%%%%%%%%%%%%%%%%%%%%%%%%%%%
\section{Reduction of the number of variables using the parametrized quantum combs with $\U(d)\times \U(d)$ symmetry}
\label{sec:variable_count}

The number of variables for our nonlinear optimization problem is summarized in \cref{table:nlopt_vars}. 
The size of our non-linear optimization problems is significantly lower than naive scaling based on the dimension of the input and output of isometries $V_i^{\lambda_i \mu_{i-1}}$ in \cref{fig:comb_implementation}.
Note that in our optimization we assume one dimensional additional memory register $\ket{\mathrm{mem}_i}$.
We compare the number of variables in our optimization with that in the naive approach to reproduce the same circuit, which is based on parametrizing each isometry $W_i$ in \cref{fig:comb_implementation} as a general isometry from $\bigoplus_{\lambda_{i-1}\in \Irr{G}^{(i-1)}, \mu_{i-1}\in\Irr{H}^{(i-1)}} V_{\lambda_{i-1}} \otimes V_{\mu_{i-1}} \otimes \CC^d$ to $\bigoplus_{\lambda_i\in \Irr{G}^{(i)}, \mu_i\in\Irr{H}^{(i)}} V_{\lambda_i} \otimes V_{\mu_i} \otimes \CC^d$ without any symmetry constraints.
The number of variables in the naive approach is given by $\sum_{i=1}^n \of[\big]{\sum_{\lambda_{i-1}, \mu_{i-1}} d_{\lambda_{i-1}} d_{\mu_{i-1}} d} \cdot \of[\big]{\sum_{\lambda_i, \mu_i} d_{\lambda_i} d_{\mu_i} d}$, where $d_\lambda$ and $d_\mu$ are dimensions of irreps $\lambda$ and $\mu$, respectively, which is summarized in \cref{table:nlopt_vars_2}.
We compare these two scalings in \cref{fig:number_of_variables}, where we can see that the number of variables in our optimization is significantly lower than that in the naive approach, especially for large $n$ and $d$.

\begin{figure}
    \centering
    \includegraphics[width=.5\linewidth]{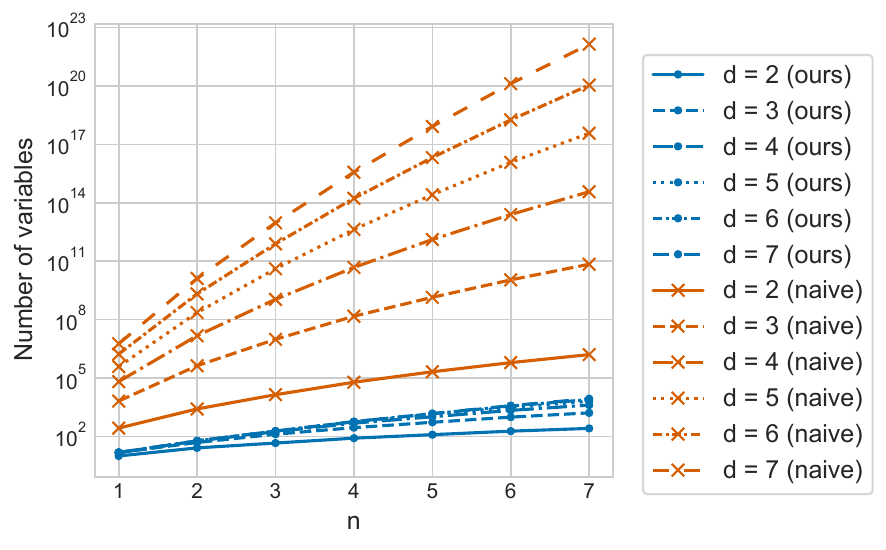}
    \caption{Comparison of the number of variables in the parametrized quantum comb for unitary transposition with $\U(d)\times \U(d)$ symmetry and that in the naive approach~\cite{mo2025parameterized} to reproduce the same circuit.
    The $x$-axis represents the query number $n$ and the $y$-axis represents the number of variables in log scale.
    The blue and orange lines represent the number of variables in the parametrized quantum comb with $\U(d)\times \U(d)$ symmetry and the naive approach, respectively.}
    \label{fig:number_of_variables}
\end{figure}

\begin{table}
    \centering
    \begin{minipage}[t]{0.3\textwidth}
        \centering
        \begin{tabular}{c|c||c|c|c|c|c|c|c}
            $f(U)$   &   \diagbox{$d$}{$n$} &   1  &  2    &  3    &   4   &   5   &  6   &  7   \\ \hline \hline
            \multirow{6}{*}{$U\tp$} &   2   &  10  &  26   &  46   &  82   &  124  & 188  & 260  \\
                                    &   3   &  15  &  49   &  131  &  284  &  536  & 987  & 1631 \\
                                    &   4   &  15  &  59   &  176  &  484  & 1026  & 2184 & 4088 \\
                                    &   5   &  15  &  59   &  189  &  556  & 1370  & 3111 & 6388 \\
                                    &   6   &  15  &  59   &  189  &  577  & 1478  & 3671 & 7869 \\
                                    &   7   &  15  &  59   &  189  &  577  & 1503  & 3830 & 8708 \\ \hline
            \multirow{6}{*}{$U\ct$} &   2   &  10  &  26   &  46   &  82   &  124  & 188  & 260  \\
                                    &   3   &  12  &  42   &  111  &  237  &  456  & 844  & 1404 \\
                                    &   4   &  12  &  46   &  135  &  363  &  781  & 1651 & 3098 \\
                                    &   5   &  12  &  46   &  142  &  403  &  982  & 2208 & 4527 \\
                                    &   6   &  12  &  46   &  142  &  415  & 1038  & 2530 & 5390 \\
                                    &   7   &  12  &  46   &  142  &  415  & 1057  & 2621 & 5873 \\
        \end{tabular}
        \caption{\textbf{Nonlinear optimization: number of variables}. Number of scalar variables in our nonlinear optimization using parametrized quantum combs with symmetries.}
        \label{table:nlopt_vars}
    \end{minipage}%
    \hfill
    \begin{minipage}[t]{0.6\textwidth}
        \centering
        \begin{tabular}{c|c||c|c|c|c|c|c|c}
            $f(U)$   &   \diagbox{$d$}{$n$} &   1  &  2    &  3    &   4   &   5   &  6   &  7   \\ \hline \hline
            \multirow{6}{*}{$U\tp$} &   2   &  $2.7\cdot 10^{2}$ & $2.6\cdot 10^{3}$ & $1.4\cdot 10^{4}$ & $6.1\cdot 10^{4}$ & $2.1\cdot 10^{5}$ & $6.2\cdot 10^{5}$ & $1.6\cdot 10^{6}$ \\
                                    &   3   &  $6.6\cdot 10^{3}$ & $4.3\cdot 10^{5}$ & $9.8\cdot 10^{6}$ & $1.5\cdot 10^{8}$ & $1.4\cdot 10^{9}$ & $1.1\cdot 10^{10}$ & $7.0\cdot 10^{10}$ \\
                                    &   4   &  $6.6\cdot 10^{4}$ & $1.5\cdot 10^{7}$ & $1.1\cdot 10^{9}$ & $4.8\cdot 10^{10}$ & $1.3\cdot 10^{12}$ & $2.5\cdot 10^{13}$ & $3.7\cdot 10^{14}$ \\
                                    &   5   &  $3.9\cdot 10^{5}$ & $2.3\cdot 10^{8}$ & $4.1\cdot 10^{10}$ & $4.3\cdot 10^{12}$ & $2.6\cdot 10^{14}$ & $1.2\cdot 10^{16}$ & $3.7\cdot 10^{17}$ \\
                                    &   6   &  $1.7\cdot 10^{6}$ & $2.1\cdot 10^{9}$ & $7.8\cdot 10^{11}$ & $1.7\cdot 10^{14}$ & $2.1\cdot 10^{16}$ & $1.8\cdot 10^{18}$ & $1.1\cdot 10^{20}$ \\
                                    &   7   &  $5.8\cdot 10^{6}$ & $1.3\cdot 10^{10}$ & $9.4\cdot 10^{12}$ & $3.7\cdot 10^{15}$ & $8.3\cdot 10^{17}$ & $1.3\cdot 10^{20}$ & $1.4\cdot 10^{22}$ \\ \hline
            \multirow{6}{*}{$U\ct$} &   2   &  $2.7\cdot 10^{2}$ & $2.6\cdot 10^{3}$ & $1.4\cdot 10^{4}$ & $6.1\cdot 10^{4}$ & $2.1\cdot 10^{5}$ & $6.2\cdot 10^{5}$ & $1.6\cdot 10^{6}$ \\
                                    &   3   &  $6.6\cdot 10^{3}$ & $2.7\cdot 10^{5}$ & $5.2\cdot 10^{6}$ & $7.0\cdot 10^{7}$ & $6.8\cdot 10^{8}$ & $5.2\cdot 10^{9}$ & $3.3\cdot 10^{10}$ \\
                                    &   4   &  $6.6\cdot 10^{4}$ & $8.0\cdot 10^{6}$ & $4.2\cdot 10^{8}$ & $1.5\cdot 10^{10}$ & $3.5\cdot 10^{11}$ & $6.4\cdot 10^{12}$ & $9.1\cdot 10^{13}$ \\
                                    &   5   &  $3.9\cdot 10^{5}$ & $1.1\cdot 10^{8}$ & $1.4\cdot 10^{10}$ & $1.1\cdot 10^{12}$ & $5.5\cdot 10^{13}$ & $2.1\cdot 10^{15}$ & $6.0\cdot 10^{16}$ \\
                                    &   6   &  $1.7\cdot 10^{6}$ & $1.0\cdot 10^{9}$ & $2.4\cdot 10^{11}$ & $3.7\cdot 10^{13}$ & $3.7\cdot 10^{15}$ & $2.7\cdot 10^{17}$ & $1.4\cdot 10^{19}$ \\
                                    &   7   &  $5.8\cdot 10^{6}$ & $6.3\cdot 10^{9}$ & $2.8\cdot 10^{12}$ & $7.7\cdot 10^{14}$ & $1.3\cdot 10^{17}$ & $1.7\cdot 10^{19}$ & $1.6\cdot 10^{21}$ \\
        \end{tabular}    
    \caption{\textbf{Naive optimization: number of variables}. Number of scalar variables in our nonlinear optimization using naive parametrized quantum combs.}
    \label{table:nlopt_vars_2}
    \end{minipage}
\end{table}

%%%%%%%%%%%%%%%%%%%%%%%%%%%%%%%%%%%%%%%%%%%%%%%%%%%%%%%%%%%%%%%%%%%%%%%%%%%%%%%%%%%%%%%%%%%%%%%%%%%%%%%%%%%%%%%%%%%%%%%%%%%%%%%%%%%%%%%%%%%%%%%%%%%%%%%%%%%%%%%%%%%%%%%%%%%%%%%%%%%%%%%%%%%%%%%%%%
\section{Proof of Lem.~\ref{lem:tensor_and_partial_trace}: Partial trace and tensor product in the commutant algebra}
\label{sec:comb_tensor_representation}
\label{subsec:proof_lem_tensor_and_partial_trace}
\begin{proof}
    The first property of Lem.~\ref{lem:tensor_and_partial_trace} is shown as follows:
    \begin{align}
        E^\mu_{rs} \otimes \1_{d_n}
        &\cong (\1_{V_\mu} \otimes \ketbra{r}{s}) \otimes \1_{d_n}\\
        &= (\1_{V_\mu} \otimes \ketbra{r}{s}) \otimes \bigoplus_{\nu\in \Irr{G}_n}\1_{V_\nu} \otimes \1_{m_{\rho_n}^{\nu}}\\
        &= \bigoplus_{\nu\in \Irr{G}_n} (\1_{V_\mu} \otimes \1_{V_\nu}\otimes \1_{m_{\rho_n}^{\nu}}) \otimes \ketbra{r}{s}\\
        &= \bigoplus_{\lambda\in \Irr{G}^{(n)}} (\1_{V_\lambda} \otimes \1_{c_{\mu\rho_n}^{\lambda}}) \otimes \ketbra{r}{s}\\
        &= \sum_{\lambda\in \Irr{G}^{(n)}} \sum_{a\in [c_{\mu\rho_n}^{\lambda}]} (\1_{V_\lambda} \otimes \ketbra{(r\xrightarrow{a}\lambda)}{(s\xrightarrow{a}\lambda)})\\
        &\cong \sum_{\lambda\in \Irr{G}^{(n)}} \sum_{p, q\in \Paths(\lambda, \scB)} \sum_{a\in [c_{\mu\rho_n}^{\lambda}]} \delta_{p, (r\xrightarrow{a}\lambda)} \delta_{q, (s\xrightarrow{a}\lambda)} E^\lambda_{pq}.
    \end{align}
    The second property of Lem.~\ref{lem:tensor_and_partial_trace} is shown as follows.
    Suppose $p$ and $q$ are given by $p = (r' \xrightarrow{a'} \lambda)$ and $q = (s' \xrightarrow{b'} \lambda)$ for $r'\in \Paths(\mu')$ and $s'\in \Paths(\mu'')$.
    Since
    \begin{align}
        \Tr_n E^\lambda_{pq}
        &=\Tr_{n}\left[\bigotimes_{i=1}^{n} \rho_i(g) E^\lambda_{pq} \bigotimes_{i=1}^{n} \rho_i(g)^\dagger\right]\\
        &= \bigotimes_{i=1}^{n-1} \rho_i(g) \Tr_n E^\lambda_{pq} \bigotimes_{i=1}^{n-1} \rho_i(g)^\dagger
    \end{align}
    holds, we obtain
    \begin{align}
        &\mu'(g) \otimes \1_{M_\mu^{(n-1)}} \Tr_n E^\lambda_{pq} \mu''(g)^{\dagger} \otimes \1_{M_\mu^{(n-1)}}\nonumber\\
        &= \Tr_n\left[\mu'(g) \otimes \1_{M_\mu^{(n-1)}} \otimes \rho_n(g) \cdot E^\lambda_{pq} \cdot \mu''(g)^{\dagger} \otimes \1_{M_\mu^{(n-1)}} \otimes \rho_n(g)^\dagger\right]\\
        &= \Tr_n\left[\bigoplus_{\lambda'\in \Irr{G}^{(n)}} \lambda'(g) \otimes \sum_{r\in \Paths(\mu', \scB)}\sum_{a\in [c_{\mu'\rho_n}^{\lambda'}]}\ketbra{r\xrightarrow{a}\lambda'}{r\xrightarrow{a}\lambda'} \cdot E^\lambda_{pq} \cdot \bigoplus_{\lambda''\in \Irr{G}^{(n)}} \lambda''(g)^\dagger \otimes \sum_{s\in \Paths(\mu'', \scB)}\sum_{b\in [c_{\mu''\rho_n}^{\lambda''}]}\ketbra{s\xrightarrow{b}\lambda''}{s\xrightarrow{b}\lambda''} \right]\\
        &= \Tr_n E^\lambda_{pq}.
    \end{align}
    Therefore, due to Schur's lemma, we obtain the following.\\
    (i) If $\mu'\neq \mu''$ holds, we obtain
    \begin{align}
        \Tr_n E^\lambda_{pq} = 0.
    \end{align}
    Since the right-hand side of Eq.~\eqref{eq:E_partial_trace} is given by 0 for this case, we obtain Eq.~\eqref{eq:E_partial_trace}.\\
    (ii) If $\mu' = \mu''$ holds, we obtain
    \begin{align}
        \Tr_n E^\lambda_{pq} \in \mathrm{span}\{E^{\mu'}_{rs} \mid r, s\in \Paths(\mu', \scB)\},
    \end{align}
    i.e.,
    \begin{align}
        \Tr_n E^\lambda_{pq} = \sum_{r, s\in \Paths(\mu', \scB)} c_{\mu' rs}^{\lambda pq} E^{\mu'}_{rs}
    \end{align}
    holds for some coefficients $c_{\mu' rs}^{\lambda pq}\in \CC$.
    We show Eq.~\eqref{eq:E_partial_trace} for the following cases.\\
    (ii-a) If $p=q$ holds, we have $r'=s'$ and $a'=b'$.
    By substituting $r=s=r'$ and $\mu=\mu'$ for Eq.~\eqref{eq:E_tensor_I}, we obtain
    \begin{align}
        E^{\mu'}_{r'r'} \otimes \1_{d_n} &= \sum_{\lambda'\in \Irr{G}^{(n)}} \sum_{p, q\in \Paths(\lambda', \scB)} \sum_{a\in [c_{\mu'\rho_n}^{\lambda'}]} \delta_{p, (r'\xrightarrow{a}\lambda')} \delta_{q, (r'\xrightarrow{a}\lambda')} E^{\lambda'}_{pq}\\
        &= \sum_{\lambda'\in \Irr{G}^{(n)}} \sum_{a\in [c_{\mu'\rho_n}^{\lambda'}]} E^{\lambda'}_{(r'\xrightarrow{a}\lambda')(r'\xrightarrow{a}\lambda')}.
    \end{align}
    By taking the partial trace of the last subsystem, we obtain
    \begin{align}
        d_n E^{\mu'}_{r'r'} &= \sum_{\lambda'\in \Irr{G}^{(n)}} \sum_{a\in [c_{\mu'\rho_n}^{\lambda'}]} \Tr_n E^{\lambda'}_{(r'\xrightarrow{a}\lambda')(r'\xrightarrow{a}\lambda')}.
    \end{align}
    Since $\Tr_n E^{\lambda'}_{(r'\xrightarrow{a}\lambda')(r'\xrightarrow{a}\lambda')}\in \mathrm{span}\{E^{\mu'}_{rs} \mid r, s\in \Paths(\mu')\}$ and $\Tr_n E^{\lambda'}_{(r'\xrightarrow{a}\lambda')(r'\xrightarrow{a}\lambda')}\geq 0$ hold, $\Tr_n E^{\lambda'}_{(r'\xrightarrow{a}\lambda')(r'\xrightarrow{a}\lambda')}$ is written as
    \begin{align}
        \Tr_n E^{\lambda'}_{(r'\xrightarrow{a}\lambda')(r'\xrightarrow{a}\lambda')} = \1_{V_{\mu'}} \otimes P_{\lambda' r' a}
    \end{align}
    using a positive semidefinite matrix $P_{\lambda' r' a}\in \End[\CC^{M_{\mu'}^{(n-1)}}]$.
    Thus, we obtain
    \begin{align}
        d_n \ketbra{r'}{r'} = \sum_{\lambda'\in \Irr{G}^{(n)}} \sum_{a\in [c_{\mu'\rho_n}^{\lambda'}]} P_{\lambda' r' a}.
    \end{align}
    Therefore, we obtain
    \begin{align}
        P_{\lambda' r' a} \propto \ketbra{r'}{r'},
    \end{align}
    i.e.,
    \begin{align}
        \Tr_n E^{\lambda'}_{(r'\xrightarrow{a}\lambda')(r'\xrightarrow{a}\lambda')} \propto E^{\mu'}_{r'r'}.
    \end{align}
    By comparing the traces of the both sides and substituting $\lambda' = \lambda$ and $a=a'$, we obtain
    \begin{align}
        \Tr_n E^{\lambda}_{(r'\xrightarrow{a'}\lambda)(r'\xrightarrow{a'}\lambda)} = {d_\lambda \over d_{\mu'}} E^{\mu'}_{r'r'},
    \end{align}
    which corresponds to Eq.~\eqref{eq:E_partial_trace}.\\
    (ii-b) If $a'=b'$ but $r'\neq s'$ hold, for $t\in \CC$, we have
    \begin{align}
        &\Tr_n E^\lambda_{(r'\xrightarrow{a'}\lambda)(r'\xrightarrow{a'}\lambda)} + t \Tr_n E^\lambda_{(r'\xrightarrow{a'}\lambda)(s'\xrightarrow{a'}\lambda)} + \overline{t}\Tr_n E^\lambda_{(s'\xrightarrow{a'}\lambda)(r'\xrightarrow{a'}\lambda)} + \abs{t}^2 \Tr_n E^\lambda_{(s'\xrightarrow{a'}\lambda)(s'\xrightarrow{a'}\lambda)}\geq 0  \nonumber\\
        &\in \mathrm{span}\{E^{\mu'}_{rs} \mid r, s\in \Paths(\mu', \scB)\},
    \end{align}
    i.e., there exists a positive semidefinite matrix $Q_{\lambda r' s' a' t}\in \End[\CC^{M_{\mu'}^{(n-1)}}]$ such that
    \begin{align}
        \Tr_n E^\lambda_{(r'\xrightarrow{a'}\lambda)(r'\xrightarrow{a'}\lambda)} + t \Tr_n E^\lambda_{(r'\xrightarrow{a'}\lambda)(s'\xrightarrow{a'}\lambda)} + \overline{t}\Tr_n E^\lambda_{(s'\xrightarrow{a'}\lambda)(r'\xrightarrow{a'}\lambda)} + \abs{t}^2 \Tr_n E^\lambda_{(s'\xrightarrow{a'}\lambda)(s'\xrightarrow{a'}\lambda)}
        &= \1_{V_{\mu'}} \otimes Q_{\lambda r' s' a' t}.
    \end{align}
    Similarly to the case (ii-a), we have
    \begin{align}
        &d_n \1_{V_{\mu'}} \otimes (\ket{r'} + t \ket{s'})(\bra{r'}+\overline{t} \bra{s'})\nonumber\\
        &\cong d_n (E^{\mu'}_{r'r'}+ tE^{\mu'}_{r's'} + \overline{t} E^{\mu'}_{s'r'} + \abs{t}^2 E^{\mu'}_{s's'})\\
        \label{eq:ii-b}
        &= \sum_{\lambda\in \Irr{G}^{(n)}} \sum_{a\in [c_{\mu'\rho_n}^{\lambda'}]} \Tr_n E^\lambda_{(r'\xrightarrow{a}\lambda)(r'\xrightarrow{a}\lambda)} + t \Tr_n E^\lambda_{(r'\xrightarrow{a}\lambda)(s'\xrightarrow{a}\lambda)} + \overline{t}\Tr_n E^\lambda_{(s'\xrightarrow{a}\lambda)(r'\xrightarrow{a}\lambda)} + \abs{t}^2 \Tr_n E^\lambda_{(s'\xrightarrow{a}\lambda)(s'\xrightarrow{a}\lambda)}\\
        &\cong \sum_{\lambda\in \Irr{G}^{(n)}} \sum_{a\in [c_{\mu'\rho_n}^{\lambda'}]} \1_{V_{\mu'}} \otimes Q_{\lambda r' s' a t}.
    \end{align}
    Therefore, we obtain
    \begin{align}
        Q_{\lambda r' s' a' t} \propto (\ket{r'} + t \ket{s'})(\bra{r'}+\overline{t} \bra{s'}),
    \end{align}
    i.e.,
    \begin{align}
        &\Tr_n E^\lambda_{(r'\xrightarrow{a'}\lambda)(r'\xrightarrow{a'}\lambda)} + t \Tr_n E^\lambda_{(r'\xrightarrow{a'}\lambda)(s'\xrightarrow{a'}\lambda)} + \overline{t}\Tr_n E^\lambda_{(s'\xrightarrow{a'}\lambda)(r'\xrightarrow{a'}\lambda)} + \abs{t}^2 \Tr_n E^\lambda_{(s'\xrightarrow{a'}\lambda)(s'\xrightarrow{a'}\lambda)}\\
        &\propto E^{\mu'}_{r'r'} + tE^{\mu'}_{r's'} + \overline{t} E^{\mu'}_{s'r'} + \abs{t}^2 E^{\mu'}_{s's'}.
    \end{align}
    Since this holds for arbitrary $t\in \CC$ and $\Tr_n E^\lambda_{(r'\xrightarrow{a'}\lambda)(r'\xrightarrow{a'}\lambda)} = {d_\lambda \over d_{\mu'}} E^{\mu'}_{r'r'}$ holds from the case (ii-a), we obtain
    \begin{align}
        \Tr_n E^\lambda_{(r'\xrightarrow{a'}\lambda)(s'\xrightarrow{a'}\lambda)} = {d_\lambda \over d_{\mu'}} E^{\mu'}_{r's'},
    \end{align}
    which corresponds to Eq.~\eqref{eq:E_partial_trace}.\\
    (ii-c) If $a'\neq b'$ holds, from Eq.~\eqref{eq:ii-b}, we obtain
    \begin{align}
        &d_n \1_{V_{\mu'}} \otimes (\ket{r'} + t \ket{s'})(\bra{r'}+\overline{t} \bra{s'})\nonumber\\
        &\cong \sum_{\lambda\in \Irr{G}^{(n)}} \sum_{a\in [c_{\mu'\rho_n}^{\lambda'}]} \Tr_n E^\lambda_{(r'\xrightarrow{a}\lambda)(r'\xrightarrow{a}\lambda)} + t \Tr_n E^\lambda_{(r'\xrightarrow{a}\lambda)(s'\xrightarrow{a}\lambda)} + \overline{t}\Tr_n E^\lambda_{(s'\xrightarrow{a}\lambda)(r'\xrightarrow{a}\lambda)} + \abs{t}^2 \Tr_n E^\lambda_{(s'\xrightarrow{a}\lambda)(s'\xrightarrow{a}\lambda)}\\
        &= \sum_{\lambda\in \Irr{G}^{(n)}} \sum_{a, b\in [c_{\mu'\rho_n}^{\lambda'}]} \sum_{k\in [c_{\mu'\rho_n}^{\lambda'}]} \omega_{c_{\mu'\rho_n}^{\lambda'}}^{(a-b)k} \nonumber\\
        &\hspace{30pt}\times \left[\Tr_n E^\lambda_{(r'\xrightarrow{a}\lambda)(r'\xrightarrow{b}\lambda)} + t \Tr_n E^\lambda_{(r'\xrightarrow{a}\lambda)(s'\xrightarrow{b}\lambda)} + \overline{t}\Tr_n E^\lambda_{(s'\xrightarrow{a}\lambda)(r'\xrightarrow{b}\lambda)} + \abs{t}^2 \Tr_n E^\lambda_{(s'\xrightarrow{a}\lambda)(s'\xrightarrow{b}\lambda)}\right],
    \end{align}
    where $\omega_r$ is the $r$-th root of unity given by $\omega_r\defeq e^{2\pi i/r}$.
    Since
    \begin{align}
        \sum_{\lambda\in \Irr{G}^{(n)}} \sum_{a, b\in [c_{\mu'\rho_n}^{\lambda'}]} \omega_{c_{\mu'\rho_n}^{\lambda'}}^{(a-b)k} \left[\Tr_n E^\lambda_{(r'\xrightarrow{a}\lambda)(r'\xrightarrow{b}\lambda)} + t \Tr_n E^\lambda_{(r'\xrightarrow{a}\lambda)(s'\xrightarrow{b}\lambda)} + \overline{t}\Tr_n E^\lambda_{(s'\xrightarrow{a}\lambda)(r'\xrightarrow{b}\lambda)} + \abs{t}^2 \Tr_n E^\lambda_{(s'\xrightarrow{a}\lambda)(s'\xrightarrow{b}\lambda)}\right]\geq 0
    \end{align}
    holds, we obtain
    \begin{align}
        &\sum_{\lambda\in \Irr{G}^{(n)}} \sum_{a, b\in [c_{\mu'\rho_n}^{\lambda'}]} \omega_{c_{\mu'\rho_n}^{\lambda'}}^{(a-b)k} \left[\Tr_n E^\lambda_{(r'\xrightarrow{a}\lambda)(r'\xrightarrow{b}\lambda)} + t \Tr_n E^\lambda_{(r'\xrightarrow{a}\lambda)(s'\xrightarrow{b}\lambda)} + \overline{t}\Tr_n E^\lambda_{(s'\xrightarrow{a}\lambda)(r'\xrightarrow{b}\lambda)} + \abs{t}^2 \Tr_n E^\lambda_{(s'\xrightarrow{a}\lambda)(s'\xrightarrow{b}\lambda)}\right]\nonumber\\
        &\propto \1_{V_{\mu'}} \otimes (\ket{r'} + t \ket{s'})(\bra{r'}+\overline{t} \bra{s'}).
    \end{align}
    Comparing the traces of the both sides, we obtain
    \begin{align}
        &\sum_{\lambda\in \Irr{G}^{(n)}} \sum_{a, b\in [c_{\mu'\rho_n}^{\lambda'}]} \omega_{c_{\mu'\rho_n}^{\lambda'}}^{(a-b)k} \left[\Tr_n E^\lambda_{(r'\xrightarrow{a}\lambda)(r'\xrightarrow{b}\lambda)} + t \Tr_n E^\lambda_{(r'\xrightarrow{a}\lambda)(s'\xrightarrow{b}\lambda)} + \overline{t}\Tr_n E^\lambda_{(s'\xrightarrow{a}\lambda)(r'\xrightarrow{b}\lambda)} + \abs{t}^2 \Tr_n E^\lambda_{(s'\xrightarrow{a}\lambda)(s'\xrightarrow{b}\lambda)}\right]\nonumber\\
        &\cong {d_\lambda \over d_{\mu'}} \1_{V_{\mu'}} \otimes (\ket{r'} + t \ket{s'})(\bra{r'}+\overline{t} \bra{s'}).
    \end{align}
    Since this holds for arbitrary $k\in [c_{\mu'\rho_n}^{\lambda'}]$, we obtain
    \begin{align}
        \Tr_n E^\lambda_{(r'\xrightarrow{a}\lambda)(r'\xrightarrow{b}\lambda)} + t \Tr_n E^\lambda_{(r'\xrightarrow{a}\lambda)(s'\xrightarrow{b}\lambda)} + \overline{t}\Tr_n E^\lambda_{(s'\xrightarrow{a}\lambda)(r'\xrightarrow{b}\lambda)} + \abs{t}^2 \Tr_n E^\lambda_{(s'\xrightarrow{a}\lambda)(s'\xrightarrow{b}\lambda)} = 0
    \end{align}
    for $a\neq b$.
    Since this holds for arbitrary $t\in \CC$, we obtain
    \begin{align}
        \Tr_n E^\lambda_{(r'\xrightarrow{a}\lambda)(s'\xrightarrow{b}\lambda)} = 0,
    \end{align}
    which corresponds to Eq.~\eqref{eq:E_partial_trace}.
\end{proof}

\end{document}